\documentclass[twocolumn,amssymb,aps,notitlepage,pra,10pt]{revtex4-2}
\pdfoutput=1

\usepackage[utf8]{inputenc}
\usepackage{natbib}
\usepackage{amsmath}
\usepackage{amsfonts}
\usepackage{amsthm}
\usepackage{qcircuit}  
\usepackage{braket}
\usepackage{braket}
\usepackage{booktabs}
\usepackage{mathdots}
\usepackage{hyperref}
\usepackage{tabularx}
\usepackage{fullpage}
\usepackage{graphicx}
\usepackage{color}
\usepackage{relsize}
\usepackage[inline]{enumitem}
\usepackage{xspace}
\usepackage[normalem]{ulem}
\usepackage[noabbrev,capitalize]{cleveref}

\theoremstyle{plain}
\newtheorem{theorem}{Theorem}

\newtheorem{proposition}[theorem]{Proposition}

\theoremstyle{definition}

\newtheorem{construction}[theorem]{Construction}

\theoremstyle{remark}
\newtheorem{remark}[theorem]{Remark}
\newtheorem*{notation*}{Notation}

\newcommand{\Z}{\mathbb{Z}}

\newcommand{\C}{\mathbb{C}}

\renewcommand{\deg}{\operatorname{deg}}

\renewcommand{\vec}[1]{\mathbf{#1}}

\newcommand{\x}{\vec{x}}

\newcommand{\etal}{{\it et al. }}
\makeatletter
\newcommand\xleftrightarrow[2][]{%
  \ext@arrow 9999{\longleftrightarrowfill@}{#1}{#2}}
\newcommand\longleftrightarrowfill@{%
  \arrowfill@\leftarrow\relbar\rightarrow}
\makeatother

\newcolumntype{C}{>{$}c<{$}}
\AtBeginDocument{
\heavyrulewidth=.08em
\lightrulewidth=.05em
\cmidrulewidth=.03em
\belowrulesep=.65ex
\belowbottomsep=0pt
\aboverulesep=.4ex
\abovetopsep=0pt
\cmidrulesep=\doublerulesep
\cmidrulekern=.5em
\defaultaddspace=.5em
}

\newlength{\localh}
\newlength{\locald}
\newbox\mybox
\def\mp#1#2{\scalebox{1}{\setbox\mybox\hbox{#2}\localh\ht\mybox\locald\dp\mybox\addtolength{\localh}{-\locald}\raisebox{-#1\localh}{\box\mybox}}}

\newcommand\showeqno{\addtocounter{equation}{1}\the\numexpr \value{equation} - 14 \relax}

\begin{document}

\title{The phase/state duality in reversible circuit design}%
\author{Matthew Amy}
\email{matt.amy@dal.ca}%
\affiliation{Department of Mathematics and Statistics \\ Dalhousie University, Halifax, Canada}

\author{Neil J. Ross}
\email{neil.jr.ross@dal.ca}
\affiliation{Department of Mathematics and Statistics \\ Dalhousie University, Halifax, Canada}%

\begin{abstract}
  The reversible implementation of classical functions accounts
  for the bulk of most known quantum algorithms. As a result,
  a number of reversible circuit constructions over the Clifford+$T$
  gate set have been developed in recent years which use both the state and 
  phase spaces, or $X$ and $Z$ bases, to reduce circuit costs beyond what is possible
  at the strictly classical level. We study and generalize two particular classes of these
  constructions: relative phase circuits, including Giles and Selinger's 
  multiply-controlled $iX$ gates and Maslov's $4$ qubit Toffoli gate,
  and measurement-assisted circuits, including Jones' Toffoli gate and 
  Gidney's temporary logical-AND.
  In doing so, we introduce general methods for
  implementing classical functions up to phase and for
  measurement-assisted termination of temporary values. 
  We then apply these
  techniques to find novel $T$-count efficient constructions of
  some classical functions in space-constrained regimes, notably multiply-controlled Toffoli gates and 
  temporary products.
\end{abstract}

\maketitle

\section{Introduction}

The reversible implementation of classical functions on a quantum
computer is crucial to many quantum algorithms, including Grover's
search algorithm \cite{Grover96} and Shor's factoring algorithm
\cite{ar:shor}. In such algorithms, oracles for classical subroutines
account for the bulk of the total circuit volume. As a result,
the optimization of quantum circuits for classical reversible
functions is central to the resource-efficient implementation of
quantum algorithms.

Two complexity measures in the design of quantum circuits are
the circuit \emph{time} or \emph{depth} and \emph{space}
or \emph{width}. 
The former corresponds roughly to the number of gates that
appear in the circuit and is sometimes weighted to account for
the fact that some gates are more costly than others, while the latter
is given by the number of qubits used in the the circuit. 
Changes to these costs are amplified 
in fault-tolerant contexts; each additional logical qubit
requires a large number of physical qubits, while longer computations
require more error correction, further increasing the physical footprint
of a quantum algorithm.

Standard
techniques for synthesizing reversible circuits can lead to
massive space overheads, 
as they rely on ancillary qubits to hold
intermediate values. This overhead can be often be
mitigated by \emph{uncomputing} intermediate values once they are no longer
useful, at the expense of extra gates.
This space-time trade-off is explored at the level of
reversible circuit synthesis through pebble games \cite{pebbles}.

Over the past decade, significant effort has been devoted to
efficiently implementing classical functions over the
Clifford+$T$ gate set, motivated by the fact that Clifford+$T$ gates
are well-suited for fault-tolerant quantum computation \cite{buslcl06}. 
In this context
the number of $T$ gates in a circuit, its \emph{$T$-count}, 
often dominates the cost in time due to the difficulty of implementing 
the $T$ gate in a fault-tolerant manner.
These recent efforts resulted in a variety of optimized Clifford+$T$
implementations of reversible gates, many of which leverage the
\emph{phase space}, or the $X$ basis, to go beyond optimizations
possibly purely in the \emph{state space}, or $Z$ basis. By
these we mean information encoded in either
the phase of a quantum state or the computational basis state,
respectively. As information in the phase and state can be freely
exchanged and independently operated on, 
we refer to this as the \emph{phase/state duality}.

In the present work we study further applications of
the phase/state duality to reversible circuit design, 
generalizing several recent constructions:
\begin{itemize}
\setlength\itemsep{-0.2em}
\item the multiply-controlled $iX$ gates of \cite{gs13} and
  \cite{s13},
\item the measurement-assisted Toffoli of \cite{j13},
\item the relative phase Toffoli-$4$ of \cite{m16}, and
\item the temporary logical-AND of \cite{g18}.
\end{itemize}
We then apply these methods to introduce new circuit designs over the
Clifford+$T$ gate set which improve the cost primarily of 
space-constrained implementations of oracles for classical functions. 

\begin{table*}
\begin{tabularx}{\textwidth}{Xclclc}
\toprule
Gate & \hspace{.5em} Ancillary state \hspace{.5em} & $T$-count \hspace{3em} & \hspace{.5em} Valid \hspace{.5em} & Notes & Ref. \\ \midrule
$U_{f\cdot g}$ & $\ket{00}$ & $2\tau(U_f) + \tau(U_g) + 8$ & -- & & \ref{circ:oracle1} \\
$U_{f\cdot g}$ & -- & $2\tau(U_f) + 2\tau(U_g) + 4$ & -- & Relative phase in the controls & \ref{circ:oracle2} \\
$U_{f\cdot g}$ & -- & $2\tau(U_f) + \tau(U_g) + 4$ & -- & Relative phase in the controls \& target & \ref{circ:oracle3} \\ \midrule
$\Lambda_k(X)$
	& $\ket{z}$ & $16(k-1)$ & $k \geq 6$ & Prior art & \cite{m16} \\
$\Lambda_k(X^\bullet)$ 
	& $\ket{z}$ & $8(k-2) + 4$ & $k\geq 2$ & Relative phase in the controls \& ancilla & \ref{circ:lambdaxbulletdirty} \\
$\Lambda_k(X)$
	& $\ket{z}$ & $16(k-2)$ & $k\geq 4$ & & \ref{circ:lambdaxdirty} \\
$\Lambda_k(X)$
	& $\ket{0}$ & $16(k-3)$ or $16(k-3) + 4$ & $k\geq 4$ & Measurement-assisted & \\
$\Lambda_k(iX)$
	& -- & $16(k-2) + 4$ & $k \geq 6$ & Prior art; Relative phase in the controls & \cite{gs13, m16} \\
$\Lambda_k(iX)$
	& -- & $16(k-3) + 4$ & $k\geq 4$ & Relative phase in the controls & \ref{circ:cix} \\
$\Lambda_k(X^\bullet)$ 
	& -- & $16(k-4) + 4$ & $k\geq 5$ & Relative phase in the controls & \ref{circ:cxbullet} \\
$\Lambda_k(X^\star)$ 
	& -- & $8(k-2)$ & $k \geq 3$ & Relative phase in the controls \& target & \ref{circ:cxstar} \\
$\Lambda_k(X^\star)$
	& $\ket{0}^{\otimes m}$ & $4m + 8(k-m-2)$ & $k \geq 5$ & Relative phase in the controls \& target \\ \midrule
$U_{f_k}$ 
	& $\ket{z}$ & $8(k-1)$ & $k \geq 2$ & \\
$U_{f_k}$ 
	& -- & $4(k-1)$ & $k \geq 2$ & Relative phase in the controls \& target & \ref{circ:fk} \\ \midrule
$3$-AND
	& $\ket{0}$ & $8$ & -- & Prior art; Relative phase in the controls & \cite{m16} \\
$3$-AND$^\dagger$
	& -- & $3$ or $4$ & -- & Relative phase; Measurement-assisted & \ref{circ:3unand} \\
$k$-AND
	& $\ket{0}$ & $16(k-3) + 4$ & $k\geq 4$ & -- & \ref{circ:kand} \\
$k$-AND$^\dagger$
	& -- & $0$ or $16(k-4) + 4$ & $k\geq 6$ & Measurement-assisted & \ref{circ:kunand} \\
$k$-AND
	& $\ket{0}$ & $8(k-2)$ & $k\geq 3$ & Relative phase in the controls & \ref{circ:andk} \\
$k$-AND$^\dagger$
	& -- & $8(k-4)$ or $8(k-4)+4$ & $k\geq 4$ & Relative phase; Measurement-assisted & \ref{circ:unkand} \\ \bottomrule
\end{tabularx}
\caption{$T$ count scaling for various reversible functions and gates.
$\tau(U_f)$ and $\tau(U_g)$ give the $T$-counts of implementations of $U_f$ and 
$U_g$, respectively. References to explicit circuits are given where possible.}
\label{tab:foo}
\end{table*}

\Cref{tab:foo} gives an overview of the novel circuits we give,
as well as the best-known constructions when possible. 
While many of these implement classical functions 
up to relative phase, all constructions with relative phases \emph{in the 
controls and/or target} can be used as drop-in replacements for matched
compute/uncompute pairs.

Our contributions include an ancilla-free $k$-control Toffoli up to 
a relative phase in both the controls and the target 
with $T$-count $8(k-2)$, improving the best known construction by 
a factor of roughly $50$\%. We also show that if this gate is used to 
initialize a temporary product
of $k$ bits as in \cite{g18}, it can be terminated with the aid of
measurement and classical control with at most $8(k-4) + 4$ $T$ gates.
Combined, these constructions give a method of temporarily 
instantiating a logical product of $k$ bits with total $T$-count at most 
$16(k-3) + 4$ and no ancillas, besides the one used to store the product.
Previous techniques require $32(k-2) + 8$ $T$ 
gates, a reduction of over $50$\% compared to the state of the art.

We also give novel space-constrained constructions for the
$k$-control Toffoli gate with reduced $T$-count 
and for the efficient multiplication of classical oracles up to phase.
We additionally show that there
exist classes of Boolean functions of degree $k$ which can be
implemented up to relative phase and without ancillas using $4(k-1)$ $T$
gates. These constructions match or improve on the
$T$-count of the best known generic method \cite{mscrd19} which uses
$O(k)$ ancillas, and have potential applications to automated and
LUT-based \cite{srwm19} synthesis of reversible circuits. More broadly,
these constructions show that there exist functions for which existing 
techniques are not able to reduce the $T$-count through the addition of 
ancillas.

\section{Background}
\label{sec:history}

\subsection{Quantum oracles}

It is well-known that, in the presence of ancillas, the gate set
\[
\{ X, \Lambda_1(X), \Lambda_2(X) \},
\]
consisting of the NOT, controlled-NOT, and Toffoli gates, is
\emph{universal for classical computing} \cite{bbcdmsssw95}. That is,
for any Boolean (or classical) function $f:\Z_2^n\rightarrow
\Z_2^m$, there exists a circuit over the gate set $\{X,
\Lambda_1(X),\Lambda_2(X)\}$ which implements a unitary $U_f$ whose
action on the computational basis is described by
\[
  \ket{\x}\ket{0\cdots 0}\ket{y} \mapsto \ket{\x}\ket{g_1(\x)\cdots
    g_k(\x)}\ket{y \oplus f(\x)}
\]
where $\x=x_1x_2,\dots,x_n$ and the $g_i$ are
 some Boolean functions $g_i:\Z_2^n\to \Z_2$.  The unitary
$U_f$ is an \emph{oracle} for $f$. The qubits beginning in the
$\ket{0}$ state are \emph{clean} ancillas. The values $g_i(x)$ used in
the process of computing $f$ are \emph{temporary values} and are often
referred to as \emph{garbage}.

To reclaim the space used for temporary values, the final result can be copied on an additional ancilla, and the circuit for $U_f$ can be run in reverse. This \emph{uncomputes} the temporary values that are no longer needed, thereby \emph{cleaning up} the garbage. This technique, colloquially known as the \emph{Bennett trick}, is shown below:
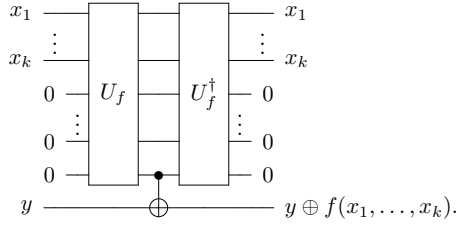
\[\hspace{-5pt}
\scalebox{0.85}{
\Qcircuit @C=.5em @R=0.6em {
\lstick{x_1} &\qw &\qw        &\qw    &\multigate{6}{U_f} &\qw      &\multigate{6}{U_f^\dagger} &\qw &\qw &\qw &\rstick{x_1} \qw \\
 & \vdots & & & & & & & & \vdots \\
\lstick{x_k} &\qw &\qw        &\qw    &\ghost{U_f} &\qw      &\ghost{U_f^\dagger} &\qw &\qw &\qw &\rstick{x_k} \qw \\
           &       &\lstick{0} &\qw &\ghost{U_f}        &\qw      &\ghost{U_f^\dagger}        &\qw &\rstick{0} \qw\\
           &       & &\vdots  &       &      &        &\vdots & \\
           &       &\lstick{0} &\qw &\ghost{U_f}        &\qw      &\ghost{U_f^\dagger}        &\qw &\rstick{0} \qw\\
           &       &\lstick{0} &\qw &\ghost{U_f}        &\ctrl{1} &\ghost{U_f^\dagger}        &\qw &\rstick{0} \qw\\
\lstick{y} &\qw &\qw        &\qw    &\qw                &\targ    &\qw                        &\qw &\qw &\qw &\rstick{y\oplus f(x_1,\dots,x_k).} \qw
}
}
\]
Note that, in the circuit above, the target of $U_f$ is in the $\ket{0}$ state. One can relax this requirement at the cost of an extra $\Lambda_1(X)$ gate:
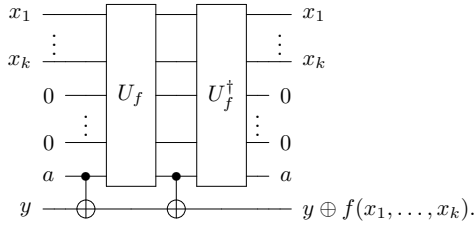
\[\hspace{-10pt}
\scalebox{0.85}{
\Qcircuit @C=.5em @R=0.6em {
\lstick{x_1} &\qw &\qw        &\qw    &\multigate{6}{U_f} &\qw      &\multigate{6}{U_f^\dagger} &\qw &\qw &\qw &\rstick{x_1} \qw \\
 & \vdots & & & & & & & & \vdots \\
\lstick{x_k} &\qw &\qw        &\qw    &\ghost{U_f} &\qw      &\ghost{U_f^\dagger} &\qw &\qw &\qw &\rstick{x_k} \qw \\
           &       &\lstick{0} &\qw &\ghost{U_f}        &\qw      &\ghost{U_f^\dagger}        &\qw &\rstick{0} \qw\\
           &       & &\vdots  &       &      &        &\vdots & \\
           &       &\lstick{0} &\qw &\ghost{U_f}        &\qw      &\ghost{U_f^\dagger}        &\qw &\rstick{0} \qw\\
           &       &\lstick{a} &\ctrl{1} &\ghost{U_f}        &\ctrl{1} &\ghost{U_f^\dagger}        &\qw &\rstick{a} \qw\\
\lstick{y} &\qw &\qw        &\targ    &\qw                &\targ    &\qw                        &\qw &\qw &\qw &\rstick{y\oplus f(x_1,\dots,x_k).} \qw
}
}
\]
In this case we say that the (uninitialized) ancilla is \emph{dirty}.

We use rounded boxes to denote oracles which do not leave any
garbage and do not modify their inputs:
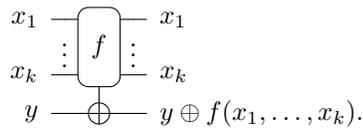
\[\hspace{-5pt}
\Qcircuit @C=.5em @R=0.6em {
\lstick{x_1} & \qw & \multimeasure{2}{f} & \qw & \rstick{x_1}\qw \\
 & \vdots & & \vdots \\
\lstick{x_k} & \qw & \ghost{f} & \qw & \rstick{x_k}\qw \\
\lstick{y} & \qw & \targ\qwx[-1] & \qw & \rstick{y\oplus f(x_1,\dots,x_k).}\qw
}
\]

\subsection{Generalized permutations}

A \emph{(unitary) generalized permutation matrix} is a permutation
matrix whose nonzero entries are elements of $\mathbb{T} = \{ z\in \C
\mid |z|=1\}$, the group of complex numbers of unit length. Every
generalized permutation matrix $U$ can be factored as the product of a
permutation matrix $P$ and a diagonal matrix $D$, i.e. $U = PD$. Note
that, since $D'=PDP^\dagger$ is also diagonal, $U$ can alternatively
be factored as
\[
U=PD=PDP^\dagger P = D'P.
\]
We will sometimes leverage this kind of quasi-commutation throughout
the remainder of this paper.
Restricting the nonzero entries of generalized permutation
matrices to $m$-th roots of unity yields the \emph{generalized
symmetric group} $\mathcal{S}(m,n)$.

\subsection{Relative phases}

Generalized permutations occur in quantum computing as
\emph{relative-phase} implementations of classical functions. In
particular, a generalized permutation $\widetilde{U_f}$ acting on the
computational basis as
\[
	\widetilde{U_f}: \ket{\x}\ket{0\cdots 0}\ket{y}\mapsto e^{ig(\x,y)}\ket{\x}\ket{0\cdots 0}\ket{y\oplus f(\x)}
\]
is called a \emph{relative-phase} implementation or oracle for $f$ and 
$e^{ig(\x, y)}$ is called the \emph{phase}. If $g(\x,y)=g(\x,y')$ for all 
$y'\in\Z_2$, we say that the phase \emph{depends only on the controls}. 
Otherwise, we say that the phase \emph{depends on the controls and the target}.

It can be observed \cite{bbcdmsssw95} that a relative phase
implementation suffices to compute any temporary value in a reversible
circuit or oracle. For example, the circuit below uses the Bennett
trick and a relative phase implementation $\widetilde{U_f}=U_fD$ of $f$ where
$D$ is some diagonal unitary to construct a \emph{phase-free} oracle
for $f$.
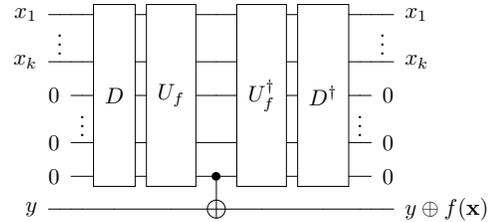
\[
\scalebox{0.85}{
\Qcircuit @C=.5em @R=0.6em {
\lstick{x_1} &\qw &\qw        &\qw    &\multigate{6}{D} &\multigate{6}{U_f} &\qw      &\multigate{6}{U_f^\dagger} &\multigate{6}{D^\dagger} &\qw &\qw &\qw &\rstick{x_1} \qw \\
	   & \vdots & & & & & & & & & & \vdots \\
\lstick{x_k} &\qw &\qw &\qw &\ghost{D}        &\ghost{U_f}        &\qw      &\ghost{U_f^\dagger}        &\ghost{D^\dagger} &\qw &\qw &\qw &\rstick{x_k} \qw \\
           &       &\lstick{0} &\qw &\ghost{D}        &\ghost{U_f}        &\qw      &\ghost{U_f^\dagger}        &\ghost{D^\dagger}        &\qw &\rstick{0} \qw\\
	   & & & \vdots & & & & & & \vdots \\
           &       &\lstick{0} &\qw &\ghost{D}        &\ghost{U_f}        &\qw      &\ghost{U_f^\dagger}        &\ghost{D^\dagger}        &\qw &\rstick{0} \qw\\
           &       &\lstick{0} &\qw &\ghost{D}        &\ghost{U_f}        &\ctrl{1} &\ghost{U_f^\dagger}        &\ghost{D^\dagger}        &\qw &\rstick{0} \qw\\
\lstick{y} &\qw &\qw        &\qw    &\qw              &\qw                &\targ    &\qw                        &\qw                      &\qw &\qw &\qw &\rstick{y\oplus f(\x)} \qw
}
}
\]
The correctness of the circuit can be established through the
quasi-commutation noted above. Indeed, we have $U_fD= D'U_f$ for some
diagonal matrix $D'$. The diagonal gates can thus be moved inwards and
cancelled, since diagonal matrices commute with controls.

More generally, an oracle $U_f$ in some compute/uncompute pair
$U_f^\dagger UU_f$
may be implemented up to a relative phase on qubit $i$
whenever the internal computation $U$ is globally constant 
on the \emph{state} space of qubit $i$. In particular,
$U$ is globally constant on the state space of the first qubit if
\[
	U(\ket{x_1}\otimes \ket{x_2\cdots x_n}) = e^{ig(\x)}\ket{x_1}\otimes U_{x_1}\ket{x_2\cdots x_n}
\]
for any $\x\in\Z_2^n$.
In practice, this accounts for the vast majority of cases where
a temporary value is computed and later uncomputed.
Additional discussion can be found in \cref{app:rphase}.

\subsection{Phase space optimizations}

The observation that the phase space can be used to optimize
reversible circuits through the use of generalized permutations dates
back to Norman Margolus \cite{ds94}. Margolus noted that the Toffoli
gate can be implemented up to a phase with just $3$ two-qubit gates,
rather than the otherwise minimal $5$. DiVincenzo and Smolin
\cite{ds94} found similar optimizations using relative phases, albeit
with a stern warning that this ``is often a dangerous thing to do.''
This idea was explored further by Barenco \etal \cite{bbcdmsssw95},
noting that implementing up to phase is generally a safe thing to do
\emph{as long as only classical computations are performed before 
uncomputing the phase}.

The idea of using the phase space to optimize reversible circuits
experienced a recent resurgence, in part due to Peter Selinger's
relative phase Clifford+$T$ implementation of the Toffoli gate. In
particular, Selinger introduced the doubly-controlled $iX$ gate
\[
	\Lambda_2(iX) : \ket{x_1}\ket{x_2}\ket{y} 
		\mapsto i^{x_1x_2}\ket{x_1}\ket{x_2}\ket{y\oplus (x_1x_2)},
\]
which can be implemented with only $4$ $T$ gates (see \cref{fig:ix})
as opposed to the optimal $7$ $T$ gates needed to implement the
Toffoli gate on the nose in the absence of ancillas and measurements
\cite{tcount}.
\begin{figure*}
\[
	\Qcircuit @C=.5em @R=0.4em @!R {
		& \ctrl{2} & \qw  \\
		& \ctrl{1} & \qw  \\
		& \gate{iX} & \qw  \\
	}
~~\mp{4.7}{=}~~
	\Qcircuit @C=.5em @R=0.4em @!R {
		& \ctrl{1} & \ctrl{2} & \qw  \\
		& \gate{S} & \ctrl{1} & \qw  \\
		& \qw & \targ & \qw  \\
	}
~~\mp{4.7}{=}~~
	\Qcircuit @C=.5em @R=0.2em @!R {
		& \qw & \qw & \qw & \qw & \ctrl{2} & \qw & \qw & \qw & \ctrl{2} & \qw & \qw & \qw  \\
		& \qw & \qw & \ctrl{1} & \qw & \qw & \qw & \ctrl{1} & \qw & \qw & \qw & \qw & \qw  \\
		& \gate{H} & \gate{T^\dagger} & \targ & \gate{T} & \targ & \gate{T^\dagger} & \targ & \gate{T} & \targ & \gate{H} & \qw  \\
	}
\]
\caption{The doubly-controlled $iX$ gate \cite{s13}.}
\label{fig:ix}
\end{figure*}
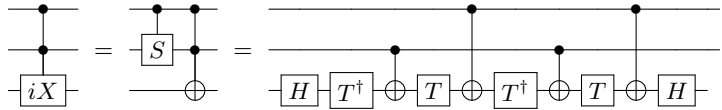
Since the erroneous phase $i^{x_1x_2}$ is irrelevant to computations in the 
state space, the doubly controlled $iX$ gate can be used interchangeably 
with a Toffoli gate to compute a temporary logical AND of two bits. When 
this temporary value is later uncomputed, the extraneous phase is also 
uncomputed, as below:
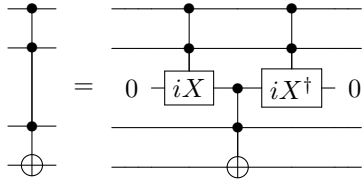
\[
	\Qcircuit @C=.5em @R=0.05em @!R {
		& \ctrl{3} & \qw  \\
		& \ctrl{2} & \qw  \\
		& \push{\rule{0em}{1.43em}}& \\
		& \ctrl{1} & \qw \\
		& \targ & \qw
	}
~~\mp{8}{=}~~
	\Qcircuit @C=.5em @R=0.05em @!R {
		& \qw & \qw & \qw & \ctrl{2} & \qw & \ctrl{2} & \qw & \qw & \qw \\
		& \qw & \qw & \qw & \ctrl{1} & \qw & \ctrl{1} & \qw & \qw & \qw \\
		& & & \lstick{0} & \gate{iX} & \ctrl{2} & \gate{iX^\dagger} & \rstick{0}\qw & & \\
		& \qw & \qw & \qw & \qw & \ctrl{1} & \qw & \qw & \qw & \qw \\
		& \qw & \qw & \qw & \qw & \targ & \qw & \qw & \qw & \qw
	}
\]
Automated methods were later developed which achieve the same 
$T$-counts by identifying the redundant $i^{x_1x_2}$ phase terms when 
regular Toffoli gates are used instead \cite{amm14,nrscm17}, mitigating 
the need for explicit relative phase constructions.

Cody Jones \cite{j13} used the $iX$ gate to implement a full Toffoli
gate using only $4$ $T$ gates, a measurement and a classically
controlled Clifford correction. The main insights were
\begin{enumerate*}[label=(\roman*)] \item that the phase $i^{x_1x_2}$
  could be corrected with a single $S^\dagger$ gate if the target of
  the $\Lambda_2(iX)$ gate is in the $\ket{0}$ state, and \item that
  the state $\ket{(x_1x_2)}$ can be uncomputed with a measurement and
  classically-controlled Clifford corrections. \end{enumerate*}
Explicitly, with a single clean ancilla, the product $\ket{(x_1x_2)}$ can
first be computed using a $\Lambda_2(iX)$ gate:
\[
	\Lambda_2(iX)\ket{x_1}\ket{x_2}\ket{0} =
        i^{x_1x_2}\ket{x_1}\ket{x_2}\ket{(x_1x_2)}.
\]
The $i^{x_1x_2}$ phase can then be immediately corrected by applying an
$S^\dagger$ gate to the ancilla. And the product $\ket{(x_1x_2)}$ can be
copied into the target register using a $\Lambda(X)$ gate and
uncomputed from the ancilla. Rather than uncomputing the state
$\ket{(x_1x_2)}$, Jones noted that it can be traded for a phase via a
Hadamard gate, since:
\[
	H\ket{(x_1x_2)} = \frac{1}{\sqrt{2}}\sum_{z\in\Z_2}(-1)^{x_1x_2z}\ket{z}.
\]
While correcting this phase with a doubly-controlled $Z$ gate would 
require $7$ $T$ gates, the ancilla can be measured first and then the 
resulting phase --- $1$ if the measurement result is $0$ or $(-1)^{x_1x_2}$ 
otherwise --- can be subsequently corrected. In the case of a measured 
value of $1$, a classically-controlled $\Lambda_1(Z)$ gate is all that is 
needed to correct to the phase. The resulting circuit is shown below.
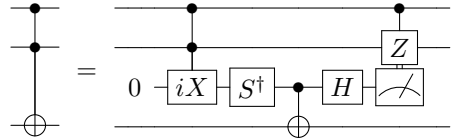
\[
	\Qcircuit @C=.5em @R=0.05em @!R {
		& \ctrl{3} & \qw  \\
		& \ctrl{2} & \qw  \\
		& \push{\rule{0em}{1.43em}} & \\
		& \targ & \qw
	}
~~\mp{6.5}{=}~~
	\Qcircuit @C=.5em @R=0.05em @!R {
		& \qw & \qw & \qw & \ctrl{2} & \qw & \qw & \qw & \ctrl{1} & \qw & \qw \\
		& \qw & \qw & \qw & \ctrl{1} & \qw & \qw & \qw & \gate{Z} & \qw & \qw \\
		& & & \lstick{0} & \gate{iX} & \gate{S^\dagger} &  \ctrl{1} & \gate{H} & \meter \cwx & \\
		& \qw & \qw & \qw & \qw & \qw & \targ & \qw & \qw & \qw & \qw
	}
\]
As Jones's circuit involves an ancilla, measurement, and a classically
controlled correction, its use in reversible circuit design remained
somewhat limited until Craig Gidney \cite{g18} observed that by
delaying the uncomputation of the temporary product $x_1x_2$, the $T$-cost
of uncomputing certain temporary values in a reversible circuit can be
reduced to $0$. Gidney introduced the \emph{temporary logical-AND}
construction by explicitly separating Jones's Toffoli into a $T$-count
$4$ circuit for initializing an ancilla with a logical AND of two bits and a
corresponding \emph{termination} circuit with $T$-count $0$. We use the
term termination, corresponding to the transformation
$\ket{x_1}\ket{x_2}\ket{(x_1x_2)} \mapsto \ket{x_1}\ket{x_2},$
to denote the fact that the circuit is non-unitary.
Both circuits are shown
below.
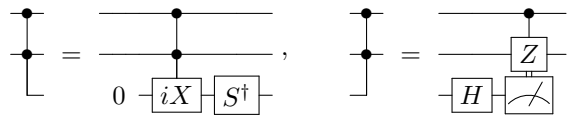
\[
	\Qcircuit @C=.5em @R=0.1em @!R {
		& \ctrl{2} & \qw  \\
		& \ctrl{1} & \qw \\
		& {} & \push{\rule{0em}{1.45em}}\qw  \\
	}
~~\mp{4.3}{=}~~
	\Qcircuit @C=.5em @R=0.1em @!R {
		& \qw & \qw & \qw & \ctrl{2} & \qw & \qw \\
		& \qw & \qw & \qw & \ctrl{1} & \qw & \qw \\
		& & & \lstick{0} & \gate{iX} & \gate{S^\dagger} & \qw \\
	}
~\mp{-15}{,}~
\qquad
	\Qcircuit @C=.5em @R=0.1em @!R {
		& \ctrl{2} & \qw  \\
		& \ctrl{1} & \qw \\
		& \qw & \push{\rule{0em}{1.45em}}{} \\
	}
~~\mp{4.3}{=}~~
	\Qcircuit @C=.5em @R=0.1em @!R {
		& \qw & \ctrl{1} & \qw \\
		& \qw & \gate{Z} & \qw \\
		& \gate{H} & \meter \cwx & \push{\rule{0em}{1.45em}} \\
	}
\]
This construction gives rise to a $k$-controlled Toffoli gate
with $4(k-1)$ $T$ gates using $k-1$ clean ancillas, as well as
implementations of classical functions $f$ with multiplicative
complexity \footnote{The minimum number of AND gates 
required to implement
$f$ over $\{$AND, XOR, NOT$\}$.} $c_{\land}(f)$
using at most $4c_{\land}(f)$ $T$ gates and
$c_{\land}(f)$ ancillas \cite{mscrd19}.
More recently, Berry \etal \cite{bgmmb19} designed a
measurement assisted termination circuit for QROM states, while
Soeken and Roetteler \cite{sr20} studied similar 
measurement assisted termination in the context of
Clifford plus arbitrary single-qubit rotations. Gidney also
explored pebble game strategies using measurement-assisted
uncomputation in \cite{g19}.

In a complementary direction, other efficient generalized permutations 
were discovered following \cite{s13}. Giles and Selinger \cite{gs13} 
gave an implementation of the multi-qubit $iX$ gate without ancillas, 
shown in \cref{fig:giles}.
\begin{figure}
\[
	\Qcircuit @C=.5em @R=.9em {
		& \qw & \ctrl{6} & \qw & \qw  \\
		& & & & \\
		& \ustick{\vdots}\qw & \ctrl{4} & \ustick{\vdots}\qw & \qw  \\
		& \qw & \ctrl{3} & \qw & \qw  \\
		& & & & \\
		& \ustick{\vdots}\qw & \ctrl{1} & \ustick{\vdots}\qw & \qw  \\
		& \qw & \gate{iX} & \qw & \qw  \\
	}
~~~\mp{9}{=}~~~
	\Qcircuit @C=.5em @R=0.9em {
		& \qw & \qw & \qw & \qw & \ctrl{6} & \qw & \qw & \qw & \ctrl{6} & \qw & \qw \\
		& & & &  & & & & & & & & & \\
		& \qw & \qw & \qw & \ustick{\vdots}\qw & \ctrl{4} & \ustick{\vdots}\qw & \qw & \ustick{\vdots}\qw & \ctrl{4} & \ustick{\vdots}\qw & \qw \\
		& \qw & \qw & \ctrl{3} & \qw & \qw & \qw & \ctrl{3} & \qw & \qw & \qw & \qw \\
		& & & & & & & & & & & & & \\
		& \qw & \ustick{\vdots}\qw & \ctrl{1} & \ustick{\vdots}\qw & \qw & \ustick{\vdots}\qw & \ctrl{1} & \ustick{\vdots}\qw & \qw & \qw & \qw \\
		& \gate{H} & \gate{T^\dagger} & \targ & \gate{T} & \targ & \gate{T^\dagger} & \targ & \gate{T} & \targ & \gate{H} & \qw
	}
\]
\caption{A circuit implementing a multiply-controlled $iX$ gate \cite{gs13}.}
\label{fig:giles}
\end{figure}
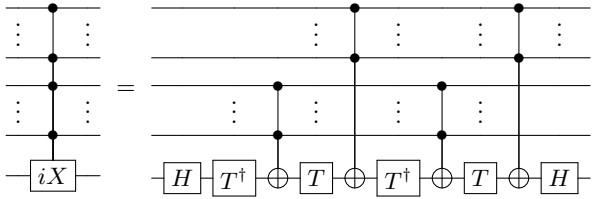
Dmitri Maslov \cite{m16} later looked at implementations of the doubly- and 
triply-controlled Toffoli gates up to other relative phases. One of Maslov's 
discoveries was a relative phase triply-controlled Toffoli gate which is shown 
in \cref{fig:maslov}. The circuit, implementing the generalized permutation
\[
	\ket{\x}\ket{y} \mapsto 
		i^{x_1x_2 + x_1x_2x_3}(-1)^{x_1x_2y}\ket{\x}\ket{y\oplus (x_1x_2x_3)},
\]
reduces the space usage to compute a product of three bits with only
$8$ $T$ gates, at the expense of a target-dependent phase of $ i^{x_1x_2 +
  x_1x_2x_3}(-1)^{x_1x_2y}$.

By using this relative phase $4$-qubit Toffoli, as well as other generalized 
permutations, novel implementations of reversible functions with reduced 
space usage were given in \cite{m16}. In the case of the Toffoli gate of 
Barenco \emph{et al.} \cite[Lemma 7.2]{bbcdmsssw95}, these techniques 
reduced the $T$-count from $12n + O(1)$ to $8n + O(1)$.
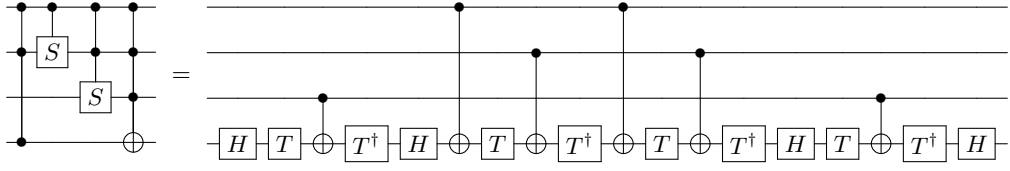
\begin{figure*}
\[
	\Qcircuit @C=.5em @R=0.58em @!R {
		& \ctrl{3} & \ctrl{1} & \ctrl{2} & \ctrl{3} & \qw \\
		& \ctrl{2} & \gate{S} & \ctrl{1} & \ctrl{2} & \qw \\
		& \qw & \qw & \gate{S} & \ctrl{1} & \qw \\
		& \ctrl{0} & \qw & \qw & \targ & \qw 
	}
	~~\mp{7.7}{=}~~
	\Qcircuit @C=.5em @R=0.4em @!R {
		& \qw & \qw & \qw & \qw & \qw & \ctrl{3} & \qw & \qw & \qw & \ctrl{3} & \qw & \qw & \qw & \qw & \qw & \qw & \qw & \qw & \qw \\
		& \qw & \qw & \qw & \qw & \qw & \qw & \qw & \ctrl{2} & \qw & \qw & \qw & \ctrl{2} & \qw & \qw & \qw & \qw & \qw & \qw & \qw \\
		& \qw & \qw & \ctrl{1} & \qw & \qw & \qw & \qw & \qw & \qw & \qw & \qw & \qw & \qw & \qw & \qw & \ctrl{1} & \qw & \qw & \qw \\
		& \gate{H} & \gate{T} & \targ & \gate{T^\dagger} & \gate{H} & \targ & \gate{T} & \targ & \gate{T^\dagger} & \targ & \gate{T}
			& \targ & \gate{T^\dagger} & \gate{H} & \gate{T} & \targ & \gate{T^\dagger} & \gate{H} & \qw 
	}
\]
\caption{A relative phase $4$-qubit Toffoli gate \cite{m16}.}
\label{fig:maslov}
\end{figure*}

\section{Circuits with ancillas}\label{sec:ancilla}

The constructions discussed in the previous section use the
phase/state duality in a variety of ways to design efficient
circuits. Typically, these circuits are then used as blackboxes: when
a more complicated functionality is required, it is reduced to a
combination of known circuits. In the remainder of the paper, we study
circuit design techniques which generalize the constructions of the
previous section and we use these techniques to define new and
efficient circuits.

Central to the techniques that we study here, and in general to the 
design of relative-phase circuits, is the
fact that the Hadamard gate induces a bijection between states of
the form $\ket{y\oplus f(\x)}$ and $(-1)^{yf(\x)}\ket{y}$. Specifically,
it is an easy calculation to show that
conjugation by Hadamard gates $\phi_H(\cdot)$ has the effect
\[
	\ket{y\oplus f(\x)}\bra{y} \xleftrightarrow{\phi_H(\cdot)}
		(-1)^{yf(\x)}\ket{y}\bra{y}
\]
Stated as circuit equalities, this is the standard fact \cite{nc00}
that a target
can be swapped for a ``control'', taken here as a $Z$ gate, and 
vice versa:
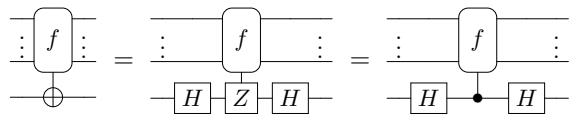
\[
	\scalebox{0.9}{
	\Qcircuit @C=.5em @R=0.3em {
  		\push{\rule{0em}{1.2em}}& \qw & \multimeasure{2}{f} & \qw & \qw \\
		& \vdots & & \vdots & \\
  		\push{\rule{0em}{1.2em}}& \qw & \ghost{f} & \qw & \qw \\
  		\push{\rule{0em}{1.2em}}& \qw & \targ \qwx[-1]  & \qw & \qw
	}}
	~~\mp{5}{=}~ 
	\scalebox{0.9}{
	\Qcircuit @C=.5em @R=0.3em {
  		\push{\rule{0em}{1.2em}}& \qw & \qw & \multimeasure{2}{f} & \qw & \qw & \qw \\
		& \vdots & & & & \vdots & \\
  		\push{\rule{0em}{1.2em}}& \qw & \qw & \ghost{f} & \qw & \qw & \qw \\
  		\push{\rule{0em}{1.2em}}& \qw & \gate{H} & \gate{Z} \qwx[-1] & \gate{H} & \qw & \qw
	}}
	~~\mp{5}{=}~
	\scalebox{0.9}{
	\Qcircuit @C=.5em @R=0.3em {
  		\push{\rule{0em}{1.2em}}& \qw & \qw & \multimeasure{2}{f} & \qw & \qw & \qw \\
		& \vdots & & & & \vdots & \\
  		\push{\rule{0em}{1.2em}}& \qw & \qw & \ghost{f} & \qw & \qw & \qw \\
  		\push{\rule{0em}{1.2em}}& \qw & \gate{H} & \ctrl{-1} & \gate{H} & \qw & \qw
	}}
\]

This simple fact can have perhaps surprising applications when oracles are
allowed to be implemented up to (relative) phase. In the first application
that we study --- to circuits with ancillas --- it provides
an alternative way to uncompute temporary values: by turning 
them into relative phases. In particular, if an ancilla initially 
in the state $\ket{a}$ is used to store a temporary value 
$\ket{a\oplus f(\x)}$, rather than uncompute this temporary value one can
simply push $f(\x)$ into the phase since 
$\phi_H(\ket{a\oplus f(\x)}\bra{a}) =  (-1)^{af(\x)}\ket{a}\bra{a}$. We
record this fact with the following statement, and illustrate its use
by implementing a $\Lambda_k(X)$ gate up to relative phase using dirty ancillas:
\begin{quote}
	\textit{A temporary value stored in a dirty ancilla can always be
	left uncomputed, at the expense of a relative phase.}
\end{quote}

\begin{construction}[Relative phase $\Lambda_k(X)$]\label{ex:tof4}
	Consider a $4$-control Toffoli gate constructed using $2$ dirty 
	ancillas, as in \cite[Lemma 7.2]{bbcdmsssw95}:
	\[
	\Qcircuit @C=.5em @R=0.5em @!R {
		& \qw & \ctrl{6} & \qw \\
		& \qw & \ctrl{5} & \qw \\
		& \qw & \ctrl{4} & \qw \\
		& \qw & \ctrl{3} & \qw \\
		& \qw & \qw & \qw \\
		& \qw & \qw & \qw \\
		& \qw & \targ & \qw
	}
	~~\mp{10.95}{=}~~
	\Qcircuit @C=.5em @R=0.5em @!R {
		& \qw & \qw & \qw & \ctrl{3} & \qw & \qw 
			& \qw & \ctrl{3} & \qw & \qw & \qw \\
		& \qw & \qw & \qw & \ctrl{3} & \qw & \qw 
			& \qw & \ctrl{3} & \qw & \qw & \qw \\
		& \qw & \qw & \ctrl{2} & \qw & \ctrl{2} 
			& \qw & \ctrl{2} & \qw & \ctrl{2} & \qw &\qw \\
		& \qw & \ctrl{2} & \qw & \qw & \qw & \ctrl{2} 
			& \qw & \qw & \qw & \qw & \qw \\
		& \qw & \qw & \ctrl{1} & \targ & \ctrl{1} 
			& \qw & \ctrl{1} & \targ & \ctrl{1} & \qw & \qw \\
		& \qw & \ctrl{1} & \targ & \qw & \targ & \ctrl{1} 
			& \targ & \qw & \targ & \qw & \qw \\
		& \qw & \targ & \qw & \qw & \qw & \targ & \qw 
			& \qw & \qw & \qw & \rstick{.} \qw
	}
	\]
	The role of the final three Toffoli gates is to uncompute the
        temporary values in the ancillas. In their absence, the
        circuit would map an input state
        $\ket{\vec{x}}\ket{a}\ket{b}\ket{y}$ to the state
        \[
        \ket{\vec{x}}\ket{a\oplus(x_1x_2)}
        \ket{b\oplus(x_1x_2x_3)}\ket{y\oplus(x_1x_2x_3x_4)}.
        \]
        Rather than perform the final Toffoli gates, we can
        conjugate the ancillas with Hadamard gates to return them to
        their initial state at the expense of a relative phase. We then get
	\[
	\scalebox{0.9}
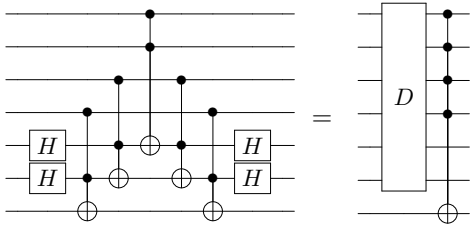
{
	\Qcircuit @C=.5em @R=0.1em @!R {
		& \qw & \qw & \qw & \qw & \ctrl{4} & \qw & \qw & \qw & \qw & \qw \\
		& \qw & \qw & \qw & \qw & \ctrl{3} & \qw & \qw & \qw & \qw & \qw \\
		& \qw & \qw & \qw & \ctrl{3} & \qw & \ctrl{3} & \qw & \qw & \qw & \qw \\
		& \qw & \qw & \ctrl{3} & \qw & \qw & \qw & \ctrl{3} & \qw & \qw & \qw \\
		& \qw & \gate{H} & \qw & \ctrl{1} & \targ & \ctrl{1} & \qw 
			& \gate{H} & \qw & \qw \\
		& \qw & \gate{H} & \ctrl{1} & \targ & \qw & \targ & \ctrl{1} 
			& \gate{H} & \qw & \qw \\
		& \qw & \qw & \targ & \qw & \qw & \qw & \targ & \qw & \qw & \qw
	}
	}
	~~\mp{11.5}{=}~~
	\scalebox{0.9}{
	\Qcircuit @C=.5em @R=0.5em @!R {
		& \qw & \multigate{5}{D} & \ctrl{6} & \qw \\
		& \qw & \ghost{D} & \ctrl{5} & \qw \\
		& \qw & \ghost{D} & \ctrl{4} & \qw \\
		& \qw & \ghost{D} & \ctrl{3} & \qw \\
		& \qw & \ghost{D} & \qw & \qw \\
		& \qw & \ghost{D} & \qw & \qw \\
		& \qw & \qw & \targ & \qw
	}
	}
	\]
        where the diagonal gate $D$ imparts a relative phase of
        $(-1)^{ax_1x_2 + bx_1x_2x_3}$.
\end{construction}

Care must be taken when trading (local) uncomputations for relative
phases, as the ancillas must remain in the same state when matched with
a (global) uncomputation later. Moreover, gates on the ancillas may
\emph{not} in general commute with such an implementation, 
while such gates do commute with an exact implementation.

One may wonder whether a similar trick can be played with circuits using 
clean ancillas. In this case, conjugation by $H$ uncomputes a temporary
value with \emph{no} relative phase:
\[
	\ket{f(\x)}\bra{0} \xleftrightarrow{\phi_H(\cdot)}
		\ket{0}\bra{0}
\]
Internally, the clean ancilla is replaced with a dirty one and the clean value is
effectively ``stored'' in the phase, to be retrieved later.
Specifically, since $H\ket{0} = \frac{1}{\sqrt{2}}\sum_{z}\ket{z}$, after the 
initial Hadamard gate the ancilla is placed in a dirty state. Adding $f(\x)$
to this ancilla results in the state 
$\frac{1}{\sqrt{2}}\sum_{z}\ket{z\oplus f(\x)}$.
Finally, since $\sum_{z}\ket{z\oplus f(x)} = \sum_{z'}\ket{z'}$ for any value 
of $f(x)$, the final Hadamard sends this state 
back to the clean state $\ket{0}$. We again summarize this fact and illustrate 
its use with a circuit construction.

\begin{quote}
	\textit{A temporary value stored in a dirty ancilla can always be
	left uncomputed without adding a relative phase by using a 
	clean ancilla.}
\end{quote}

\begin{construction}[Phase-based Bennett]
	Recall Bennett's compute-copy-uncompute scheme to reclaim
        temporary ancillas:
	
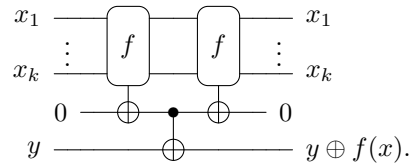
\[
	\Qcircuit @C=.5em @R=0.6em {
		\lstick{x_1} & \qw & \qw & \qw & \multimeasure{2}{f} & \qw 
			& \multimeasure{2}{f} & \qw & \qw & \qw &\rstick{x_1} \qw \\
		& \vdots & & & & & & & & \vdots & \\
		\lstick{x_k} & \qw & \qw & \qw & \ghost{f} & \qw & \ghost{f} 
			& \qw & \qw & \qw &\rstick{x_k} \qw \\
		& & \lstick{0} & \qw & \targ \qwx[-1]& \ctrl{1} & \targ \qwx[-1] 
			& \qw & \rstick{0} \qw \\
		\lstick{y} & \qw & \qw & \qw & \qw & \targ & \qw & \qw & \qw & \qw 
			& \rstick{y\oplus f(x).} \qw
	}
	\]
	The compute-copy-uncompute construction can be reduced to a
        single compute by using the phase space to temporarily store the
	ancilla's (clean) value. Specifically, by applying a Hadamard gate to 
	the ancilla, the clean $\ket{0}$ state is swapped into the phase
	space. This clean phase can later be swapped back into the 
	state space with a Hadamard gate, uncomputing the intermediate state
	as below.
	
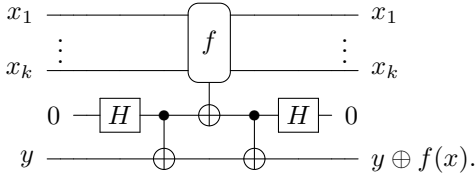
\[
	\Qcircuit @C=.5em @R=0.6em {
		\lstick{x_1} & \qw & \qw & \qw & \qw & \qw & \multimeasure{2}{f} 
			& \qw & \qw & \qw & \qw & \rstick{x_1} \qw \\
		& \vdots & & & & & & & & & \vdots & \\
		\lstick{x_k} & \qw & \qw & \qw & \qw & \qw & \ghost{f} & \qw 
			& \qw & \qw & \qw & \rstick{x_k} \qw \\
		& & \lstick{0} & \qw & \gate{H} & \ctrl{1} & \targ \qwx[-1] & \ctrl{1} 
			& \gate{H} & \rstick{0} \qw \\
		\lstick{y} & \qw & \qw & \qw & \qw & \targ & \qw & \targ & \qw & \qw 
			& \qw & \rstick{y\oplus f(x).}\qw
	}
	\]
	An additional $\Lambda_1(X)$ gate is needed, as after swapping 
	the ancilla's initial (clean) state into the phase the ancilla exists in a \emph{dirty}
	state. This dirty value is then added to the target register twice, canceling out, 
	as in the constructions of \cite{bbcdmsssw95}.
\end{construction}

To use the phase-based Bennett trick above for cleanup, the compute circuit 
needs to be designed to work with dirty ancillas --- this can be prohibitive as 
implementations 
of oracles using clean ancillas are typically more time-efficient. However, in 
cases where the temporary values are expensive to compute it can sometimes 
be advantageous to adjust the inner computation to work with dirty ancillas, 
as the following construction shows.

\begin{construction}[Oracle multiplication]
	Consider the circuit below multiplying two Boolean 
	functions $f$ and $g$ using two clean ancillas:
	{\small
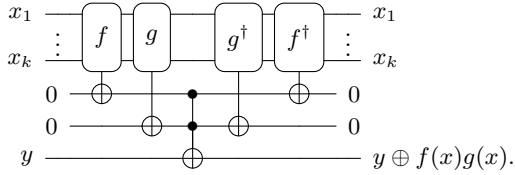
	\[
	\hspace{-2.5em}
	\Qcircuit @C=.5em @R=0.5em {
		\lstick{x_1} & \qw & \qw & \multimeasure{2}{f} & \multimeasure{2}{g} 
			& \qw & \multimeasure{2}{g^\dagger} & \multimeasure{2}{f^\dagger} 
			& \qw & \qw & \rstick{x_1} \qw \\
		& \vdots & & & & & & & & \vdots & \\
		\lstick{x_k}& \qw & \qw & \ghost{f} & \ghost{g} & \qw 
			& \ghost{g^\dagger} & \ghost{f^\dagger} & \qw & \qw 
			& \rstick{x_k} \qw \\
		& & \lstick{0} & \targ \qwx[-1] & \qw & \ctrl{2} & \qw & \targ \qwx[-1]  
			& \rstick{0} \qw \\
		& & \lstick{0} & \qw & \targ \qwx[-2]  & \ctrl{1} & \targ \qwx[-2]  
			& \qw & \rstick{0} \qw \\
		\lstick{y} & \qw & \qw & \qw & \qw & \targ & \qw & \qw & \qw 
			& \qw & \rstick{y\oplus f(x)g(x).} \qw
	}
	\]
	}
	We can eliminate one uncomputation of either $f$ or $g$ by swapping the 
	phase and state space of the corresponding ancilla, at the expense of one 
	extra Toffoli gate to deal with the now dirty ancilla.
	{\small
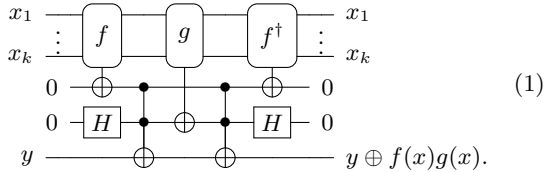
	\[
	\begin{array}{cc}
	\Qcircuit @C=.5em @R=0.4em {
		\lstick{x_1}& \qw & \qw & \multimeasure{2}{f} & \qw & \multimeasure{2}{g} 
			& \qw & \multimeasure{2}{f^\dagger} & \qw & \qw & \rstick{x_1} \qw \\
		& \vdots & & & & & & & & \vdots & \\
		\lstick{x_k}& \qw & \qw & \ghost{f} & \qw & \ghost{g} & \qw 
			& \ghost{f^\dagger} & \qw & \qw & \rstick{x_k} \qw \\
		& & \lstick{0} & \targ \qwx[-1] & \ctrl{2} & \qw & \ctrl{2} 
			& \targ \qwx[-1] & \rstick{0} \qw \\
		& & \lstick{0} & \gate{H} & \ctrl{1} & \targ \qwx[-2] & \ctrl{1} 
			& \gate{H} & \rstick{0} \qw \\
		\lstick{y}& \qw & \qw & \qw & \targ & \qw & \targ & \qw & \qw & \qw 
			& \rstick{y\oplus f(x)g(x).} \qw
	} 
	\qquad \qquad\qquad & \tag{\showeqno} \label{circ:oracle1}
	\end{array}
	\]
	} 
	Note that here, the two Toffoli gates can be replaced with
        appropriate $\Lambda_2(iZ)$ and Hadamard gates, requiring $8$
        $T$ gates rather than the $14$ that would be required if
        Toffoli gates were used.
\end{construction}

\begin{proposition}\label{circ:bmult}
	Let $f,g:\Z_2^k \to \Z_2$ be Boolean functions and suppose
	the oracles $U_f$ and $U_g$ can be implemented with $T$-count
	$\tau(U_f)$ and $\tau(U_g)$, respectively.
	With two additional clean ancillas, the
        oracle $U_{f\cdot g}$ can be implemented by a circuit of
        $T$-count $2\tau (U_f) + \tau (U_g) + 8$.
\end{proposition}

We now apply these observations to construct implementations 
of multiply-controlled Toffoli gates with a single dirty ancilla, both 
up to relative phase and implemented exactly, 
using fewer $T$ gates than previously known.

\begin{construction}[$\Lambda_k(X)$ with a single dirty ancilla]\label{ex:cxbulletdirty}
	Recall \cite[Lemma 7.3]{bbcdmsssw95} that
	a $k$-controlled Toffoli gate 
	can be decomposed as follows, using a single dirty ancilla:
	
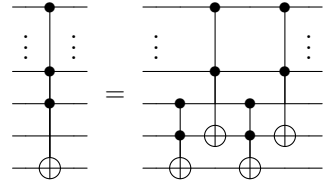
\[
	\Qcircuit @C=.5em @R=0.4em @!R {
		& \qw & \ctrl{5} & \qw & \qw \\
		& \vdots & & \vdots & \\
		& \qw & \ctrl{3} & \qw & \qw \\
		& \qw & \ctrl{2} & \qw & \qw \\
		& \qw & \qw & \qw & \qw \\
		& \qw & \targ & \qw & \qw
	}
	~~\mp{9}{=}~~
	\Qcircuit @C=.5em @R=0.4em @!R {
		& \qw & \qw & \ctrl{3} & \qw 
			& \ctrl{3} & \qw & \qw \\
		& \vdots & & & & & \vdots & \\
		& \qw & \qw & \ctrl{2} 
			& \qw & \ctrl{2} & \qw & \qw \\
		& \qw & \ctrl{2} & \qw & \ctrl{2} 
			& \qw & \qw & \qw \\
		& \qw & \ctrl{1} & \targ & \ctrl{1} 
			& \targ & \qw & \qw \\
		& \qw & \targ & \qw & \targ & \qw 
			& \qw & \rstick{.} \qw
	}
	\]
	The construction can be applied recursively using the
	$k$th bit as a dirty ancilla. However, this results in
	an exponential gate count, since the dirty ancilla needs to be cleaned
	in each recursive instantiation.

	We can recover a linear gate count with a simple recursive
        construction by trading the
        temporary value held in the dirty ancilla for a phase in each
        step. We then obtain the following equality, where $D$ and
        $D'$ are some particular diagonal gates:
	{\small
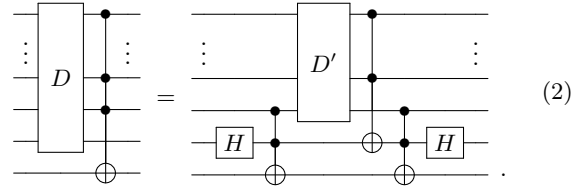
	\[
	\begin{array}{cc}
	\Qcircuit @C=.5em @R=0.1em @!R {
		& \qw & \multigate{4}{D} & \ctrl{5} & \qw & \qw \\
		& \vdots & & & \vdots & \\
		& \qw & \ghost{D} & \ctrl{3} & \qw & \qw \\
		& \qw & \ghost{D} & \ctrl{2} & \qw & \qw \\
		\push{\rule{0em}{1.2em}} & \qw & \ghost{D} & \qw & \qw & \qw \\
		& \qw & \qw & \targ & \qw & \qw
	}
	~~\mp{9.5}{=}~~
	\Qcircuit @C=.5em @R=0.05em @!R {
		& \qw & \qw & \qw & \multigate{3}{D'} & \ctrl{3} & \qw 
			& \qw & \qw & \qw \\
		& \vdots & & & & & & & \vdots & \\
		& \qw & \qw & \qw & \ghost{D'} & \ctrl{2} 
			& \qw & \qw & \qw & \qw \\
		& \qw & \qw & \ctrl{2} & \ghost{D'} & \qw & \ctrl{2} 
			& \qw & \qw & \qw \\
		& \qw & \gate{H} & \ctrl{1} & \qw & \targ & \ctrl{1} 
			& \gate{H} & \qw & \qw \\
		& \qw & \qw & \targ & \qw & \qw & \targ & \qw 
			& \qw & \rstick{.} \qw
	}
	& \tag{\showeqno}\label{circ:lambdaxbulletdirty}
	\end{array}
	\]
	}
	The precise form of $D$ is given in \cref{app:proof}.

	The above construction reduces the $T$-count for a
        $k$-controlled Toffoli gate with a single dirty ancilla to
        $8(k-2) + 4$, using $\Lambda_2(iX)$ gates to implement
        Toffolis, at the expense of a relative phase on the controls
        and ancilla. By comparison, the best-known $T$-count for
        $\Lambda_k(X)$ with a single dirty ancilla is $16(k-1)$ using
        \cite[Lemma 7.3]{bbcdmsssw95} together with Maslov's
        circuit \cite{m16} for the inner $\Lambda_{k/2}(X)$ gates.
	
	A $k$-controlled Toffoli gate can also be implemented on the
        nose by performing a final overall cleanup of the dirty ancilla, as
        below:
	\[
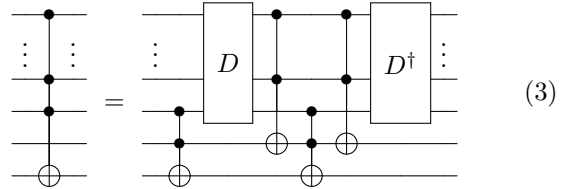
	\begin{array}{cc}
	\Qcircuit @C=.5em @R=0.4em @!R {
		& \qw & \ctrl{5} & \qw & \qw \\
		& \vdots & & \vdots & \\
		\push{\rule{0em}{.5em}} & \qw & \ctrl{3} & \qw & \qw \\
		& \qw & \ctrl{2} & \qw & \qw \\
		& \qw & \qw & \qw & \qw \\
		& \qw & \targ & \qw & \qw
	}
	~~\mp{9}{=}~~
	\Qcircuit @C=.5em @R=0.32em @!R {
		& \qw & \qw & \multigate{3}{D} & \ctrl{3} & \qw 
			& \ctrl{3} & \multigate{3}{D^\dagger} & \qw  & \qw \\
		& \vdots & & & & & & & \vdots & \\
		& \qw & \qw & \ghost{D} & \ctrl{2} 
			& \qw & \ctrl{2} & \ghost{D^\dagger} & \qw & \qw \\
		& \qw & \ctrl{2} & \ghost{D} & \qw & \ctrl{2} 
			& \qw & \ghost{D^\dagger} & \qw & \qw \\
		& \qw & \ctrl{1} & \qw & \targ & \ctrl{1} 
			& \targ & \qw & \qw & \qw \\
		& \qw & \targ & \qw & \qw & \targ & \qw 
			& \qw & \qw & \qw
	}
	& \tag{\showeqno}\label{circ:lambdaxdirty}
	\end{array}
	\]
	For $k\geq 4$, this gives a $\Lambda_k(X)$ gate with a single
        dirty ancilla and $T$-count $16(k-2)$. This reduces the
        $T$-count of the best-known construction by $16$.
\end{construction}

\begin{proposition}\label{prop:dirtyancilla}
	Let $k\in\Z^{\geq 4}$. With a single dirty ancilla, the
        $\Lambda_k(X)$ gate can be implemented by a circuit of
        $T$-count $16(k-2)$.
\end{proposition}

\begin{proposition}
        Let $k\in\Z^{\geq 2}$. With a single dirty ancilla, the
        $\Lambda_k(X)$ gate can be implemented up to a phase in the
        controls and the ancilla by a circuit of $T$-count $8(k-2)+4$.
\end{proposition}

\begin{remark}
	In general it is preferable to use an alternate form of
        the $\Lambda_k(X)$ where the ancilla is cleaned by pushing the
        temporary value into the phase, followed by a phase cleanup.
	In this case, if the $\Lambda_k(X)$ is later uncomputed the
	phase cleanup can be cancelled by automated means.
	The circuit is shown below:
	{\small
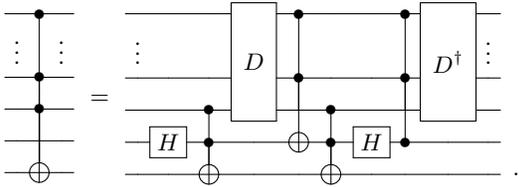
	\[
	\Qcircuit @C=.5em @R=0.1em @!R {
		& \qw & \ctrl{5} & \qw & \qw \\
		& \vdots & & \vdots & \\
		\push{\rule{0em}{1.2em}} & \qw & \ctrl{3} & \qw & \qw \\
		& \qw & \ctrl{2} & \qw & \qw \\
		& \qw & \qw & \qw & \qw \\
		& \qw & \targ & \qw & \qw
	}
	~~\mp{9.5}{=}~~
	\Qcircuit @C=.5em @R=0.05em @!R {
		& \qw & \qw & \qw & \multigate{3}{D} & \ctrl{3} & \qw 
			& \qw & \ctrl{3} & \multigate{3}{D^\dagger} & \qw  & \qw \\
		& \vdots & & & & & &  & & & \vdots & \\
		& \qw & \qw & \qw & \ghost{D} & \ctrl{2} 
			& \qw & \qw & \ctrl{2} & \ghost{D^\dagger} & \qw & \qw \\
		& \qw & \qw & \ctrl{2} & \ghost{D} & \qw & \ctrl{2} 
			& \qw & \qw & \ghost{D^\dagger} & \qw & \qw \\
		& \qw & \gate{H} & \ctrl{1} & \qw & \targ & \ctrl{1} 
			& \gate{H} & \ctrl{0} & \qw & \qw & \qw \\
		& \qw & \qw & \targ & \qw & \qw & \targ & \qw 
			& \qw & \qw & \qw & \rstick{.} \qw
	}
	\]
	}
\end{remark}

\section{Ancilla-free circuits}\label{sec:ancillafree}

We now turn our attention to even more space-efficient constructions
and in particular to circuits which do not use ancillas. We cover
design techniques that can be used to implement generalized
permutations. We then use these techniques to give $T$-count efficient
relative phase implementations of multiply-controlled Toffoli gates.
We also identify Boolean functions for which the $T$-count of the 
best-known construction \cite{mscrd19} can be matched or beaten
without the use of ancillas and measurement (but at the cost of a
relative phase).

Recall that $\mathbb{D}[\omega]$ is the ring of dyadic fractions
extended by $\omega=e^{i\pi/4}$. It was shown in \cite{gs13} that, for
$n\geq 4$, every $n$-qubit unitary with determinant 1 and entries in
$\mathbb{D}[\omega]$ can be exactly represented by an ancilla-free
Clifford+$T$ circuit. Since the determinant of a permutation is either
$1$ or $-1$, depending on whether the permutation is even or odd, it
follows that exactly the even permutations can be represented by an
ancilla-free Clifford+$T$ circuit. If we allow relative phase
implementations, however, any permutation can be implemented.

To find generalized permutations with efficient ancilla-free implementations,
it can be helpful to \emph{push all computation to the phase space}.
In particular, given a relative-phase implementation (suppressing the
constant $\ket{\x}$ register)
\[
	\widetilde{U_f} = e^{ig(\x)}\ket{y\oplus f(\x)}\bra{y},
\]
conjugating by $H$ pushes all computation to the phase:
\[
	H\widetilde{U_f}H = e^{ig(\x)}(-1)^{yf(\x)}\ket{y}\bra{y}.
\]
By instead synthesizing the \emph{phase oracle}, we can more easily find
relative phases $e^{ig(\x)}$ which reduce the overall $T$-count, as the
next constructions show.

\begin{construction}[Relative phase $\Lambda_k(Z)$]
	The circuit from \cref{ex:cxbulletdirty} can be equivalently derived 
	by synthesizing the phase-space version of $\Lambda_k(Z)$ up to
	relative phase.
	In particular, to perform a multiply-controlled $Z$
	gate up to relative phase, we want to compute some phase
	\[
		(-1)^{y\cdot x_1\cdots x_k}e^{ig(\x,a)}
	\]
	where $y$ is the target and $a$ is an ancillary bit.

	We can build a circuit to do so by first multiplying (in $\mathbb{F}_2$) 
	$y$ by $x_k$ and adding this to an ancilla $a$, 
	then applying a phase of $(-1)^{x_1\cdots x_{k-1}(a + yx_k)}$. In
	particular we have the equality below, for some diagonal gate $D$:
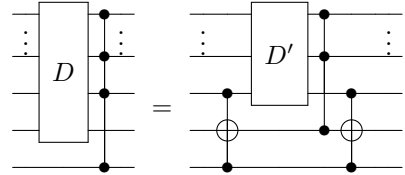
	\[
	\Qcircuit @C=.5em @R=0.2em {
		\push{\rule{0em}{1.2em}} & \qw & \multigate{4}{D} & \ctrl{5} & \qw & \qw \\
		& \vdots & & & \vdots & \\
		\push{\rule{0em}{1.2em}} & \qw & \ghost{D} & \ctrl{3} & \qw & \qw \\
		\push{\rule{0em}{1.2em}} & \qw & \ghost{D} & \ctrl{2} & \qw & \qw \\
		\push{\rule{0em}{1.2em}} & \qw & \ghost{D} & \qw & \qw & \qw \\
		\push{\rule{0em}{1.2em}} & \qw & \qw & \ctrl{0} & \qw & \qw
	}
	~~\mp{9.7}{=}~~
	\Qcircuit @C=.5em @R=0.2em {
		\push{\rule{0em}{1.2em}} & \qw & \qw & \multigate{3}{D'} & \ctrl{3} & \qw 
			& \qw & \qw & \qw \\
		& \vdots & & & & & & \vdots & \\
		\push{\rule{0em}{1.2em}} & \qw & \qw & \ghost{D'} & \ctrl{2} 
			& \qw & \qw & \qw & \qw \\
		\push{\rule{0em}{1.2em}} & \qw & \ctrl{2} & \ghost{D'} & \qw & \ctrl{2} 
			& \qw & \qw & \qw \\
		\push{\rule{0em}{1.2em}} & \qw & \targ & \qw & \ctrl{0} & \targ 
			& \qw & \qw & \qw \\
		\push{\rule{0em}{1.2em}} & \qw & \ctrl{-1} & \qw & \qw & \ctrl{-1} 
			& \qw & \qw & \qw
	}
	\]
	Conjugating by $H$ on the target gives the $\Lambda_k(X)$ circuit from
	\cref{ex:cxbulletdirty}, up to swapping controls and targets by commuting 
	the Hadamards.
\end{construction}

\begin{construction}[Selinger's $\Lambda_2(iX)$]
	The doubly-controlled $iX$ gate in \cref{fig:ix} computes the
        following transformation on computational basis states:
	\[
		\ket{x_1}\ket{x_2}\ket{y} \mapsto i^{x_1x_2}\ket{x_1}\ket{x_2}\ket{y\oplus (x_1x_2)}.
	\]
	To see how $\Lambda_2(iX)$ arises naturally as an efficient relative phase 
	implementation of the Toffoli gate over Clifford+$T$, it is helpful to consider
	the doubly-controlled $Z$ gate:
	\[
		\ket{x_1}\ket{x_2}\ket{y} \mapsto (-1)^{x_1x_2y}\ket{x_1}\ket{x_2}\ket{y}.
	\]
	Since $(-1)^{x_1x_2y}=\omega^{4x_1x_2y}$, we can use the equality 
	\cite{s13}
	\begin{align*}
		4x_1x_2y =\; &x_1 + x_2 + y - (x_1\oplus x_2) - (x_1\oplus y) \\
			&- (x_2 \oplus y) + (x_1\oplus x_2 \oplus y)
	\end{align*}
	to implement the doubly-controlled $Z$ gate over Clifford+$T$
        by computing each of the terms in the above sum (using
        $\Lambda_1(X)$ gates) and applying a $T$ or a $T^\dagger$
        gate. If we only apply
        the $4$ rotations which depend on $y$, and noting that $2x_1x_2 =
        x_1 + x_2 - (x_1\oplus x_2),$ the resulting phase term is
	\[
		4x_1x_2y - 2x_1x_2 = y - (y\oplus x_1) - (y\oplus x_2) + (y\oplus x_1\oplus x_2)
	\]
	Conjugating by $H$ on the target then sends the output 
	$(-1)^{x_1x_2y}i^{-x_1x_2}\ket{x_1}\ket{x_2}\ket{y}$ to
	$i^{-x_1x_2}\ket{x_1}\ket{x_2}\ket{y\oplus (x_1x_2)}$,
	implementing the desired transformation up to a phase of 
	$i^{-x_1x_2}$.
\end{construction}

In general, for any classical function $f:\Z_2^n\rightarrow \Z_2$, 
an oracle for $f$ can be implemented up to relative phase by taking the
Fourier transform \cite{aam17} of $yf(x)$ and then dropping 
all terms that do not involve $y$. Explicitly, if 
\[
	(-1)^{yf(x)} = \omega^{g(x) + yh(x)},
\]
then it suffices to implement the phase oracle $\omega^{yh(x)}$.
We summarize this in the following statement:
\begin{quote}
	\textit{A relative-phase implementation of $U_f$ can be obtained by
	taking the Fourier transform of $yf(\x)$ and truncating all terms
	which do not depend on $y$.}
\end{quote}

Using this idea, we can devise a method to multiple two Boolean functions 
$f,g$ up to relative phase without ancillas, given oracles for $f$ and $g$. In 
particular, since
\[\scalebox{0.9}{$
	(-1)^{y\cdot f(x)g(x)}i^{-f(x)g(x)} 
		= \omega^{y - y\oplus f(x) - y \oplus g(x) + 
		y \oplus f(x) \oplus g(x)},
$}\]
we can alternately add $f(x)$ and $g(x)$ into the target $y$ and apply the
appropriate $T$ or $T^\dagger$ gate, as below:
\[
	\begin{array}{cc}
	\hspace{-.7em}
	\scalebox{0.7}{
	\Qcircuit @C=.5em @R=0.5em {
		& \qw & \multigate{2}{D} & \multimeasure{2}{f\cdot g} & \qw & \qw \\
		& \vdots & & & \vdots & \\
		& \qw & \ghost{D} & \ghost{f\cdot g} & \qw & \qw\\
		\push{\rule{0em}{1.4em}}& \qw & \qw & \targ \qwx[-1] & \qw & \qw
	}
	}
	~\mp{4.2}{=}
	\scalebox{0.7}{
	\Qcircuit @C=.25em @R=0.5em {
		& \qw & \qw & \multimeasure{2}{f} & \qw & \multimeasure{2}{g} & \qw & \multimeasure{2}{f} & \qw 
			& \multimeasure{2}{g} & \qw & \qw \\
		& \vdots & & & & & & & & & \vdots & \\
		& \qw & \qw & \ghost{f} & \qw & \ghost{g} & \qw & \ghost{f} & \qw 
			& \ghost{g} & \qw & \qw\\
		& \gate{H} & \gate{T} & \targ \qwx[-1] & \gate{T^\dagger} & \targ \qwx[-1] 
			& \gate{T} & \targ \qwx[-1] & \gate{T^\dagger} & \targ \qwx[-1] & \gate{H} 
			& \qw 
	}
	}
	& 
	\hspace{-.7em}\tag{\showeqno}\label{circ:oracle2}
	\end{array}
\]

The above \emph{matched} multiplication construction generalizes the
multiply-controlled $iX$ implementation of Giles and Selinger \cite{gs13}.

We arrive at a slightly different form of ancilla-free oracle multiplication
by noting that the final computation of $U_g$ serves only to uncompute the
temporary value $g(x)$. As noted in \cref{sec:ancilla}, since the target
is conjugated with Hadamard gates, this temporary value can instead be swapped
into the phase. The result is the \emph{unmatched} oracle multiplication circuit
below, which generalizes Maslov's relative phase $4$-qubit Toffoli \cite{m16}. Note
that unlike matched multiplication, unmatched multiplication results in a 
target-dependent phase.
\[
	\begin{array}{cc}
	\scalebox{0.7}{
	\Qcircuit @C=.5em @R=0.5em {
		& \qw & \multigate{3}{D} & \multimeasure{2}{f\cdot g} & \qw & \qw \\
		& \vdots & & & \vdots & \\
		& \qw & \ghost{D} & \ghost{f\cdot g} & \qw & \qw\\
		\push{\rule{0em}{1.4em}} & \qw  & \ghost{D} & \targ \qwx[-1] & \qw & \qw
	}
	}
	~\mp{4.2}{=}
	\scalebox{0.7}{
	\Qcircuit @C=.3em @R=0.5em {
		& \qw & \qw & \multimeasure{2}{f} & \qw & \multimeasure{2}{g} & \qw & \multimeasure{2}{f} & \qw & \qw & \qw \\
		& \vdots & & & & & & & & \vdots & \\
		& \qw & \qw & \ghost{f} & \qw & \ghost{g} & \qw &\ghost{f} & \qw & \qw & \qw\\
		& \gate{H} & \gate{T} & \targ \qwx[-1] & \gate{T^\dagger} & \targ \qwx[-1] 
			& \gate{T} & \targ \qwx[-1] & \gate{T^\dagger} & \gate{H} & \qw 
	}
	}
	& \tag{\showeqno}\label{circ:oracle3}
	\end{array}
\]

\begin{proposition}\label{circ:bmultancfreec}
  	Let $f,g:\Z_2^k \to \Z_2$ be Boolean functions and suppose the oracles
	$U_f$ and $U_g$ can be implemented with $T$-count $\tau(U_f)$ and
	$\tau(U_g)$, respectively.
	With no additional ancillas, the oracle
        $U_{f\cdot g}$ can be implemented up to a phase in the
        controls by a circuit of $T$-count $2\tau(U_f) + 2\tau(U_g) +
        4$.
\end{proposition}

\begin{proposition}\label{circ:bmultancfreect}
  	Let $f,g:\Z_2^k \to \Z_2$ be Boolean functions and suppose the oracles
	$U_f$ and $U_g$ can be implemented with $T$-count $\tau(U_f)$ and
	$\tau(U_g)$, respectively. With no additional ancillas, the oracle
        $U_{f\cdot g}$ can be implemented up to a phase in the
        controls and the target by a circuit of $T$-count $2\tau(U_f)
        + \tau(U_g) + 4$.
\end{proposition}

The multiplication constructions above can be instantiated in various ways to 
design relative phase circuits without ancillas. We now cover some of these 
applications.

\begin{construction}[Efficient high-degree oracles]\label{ex:efficient}
By recursively instantiating $U_g$ in the unmatched multiplication, we can 
quickly (in the $T$-count)
grow the degree of a Boolean function by setting $f(x)=x$, multiplying in
one control at each iteration.
\[
\begin{array}{cc}
\hspace{-1em}
\scalebox{0.65}{
	\Qcircuit @C=.5em @R=.65em @!R {
		& \qw & \multigate{4}{D} & \multimeasure{3}{f_k} & \qw & \qw \\
		& \vdots & & & \vdots & \\
		& \qw & \ghost{D} & \ghost{f_k} & \qw & \qw  \\
		& \qw & \ghost{D} & \ghost{f_k} & \qw & \qw \\
		& \qw & \ghost{D} & \targ \qwx[-1] & \qw & \qw
	}
}
~\mp{6}{=}
\scalebox{0.65}{
	\Qcircuit @C=.5em @R=0.1em @!R {
		& \qw & \qw & \qw & \qw & \multigate{4}{D} & \multimeasure{2}{f_{k-1}} & \qw 
			& \qw & \qw & \qw & \qw \\
		& \vdots & & & & & & & & & \vdots \\
		& \qw & \qw & \qw & \qw & \ghost{D} & \ghost{f_{k-1}} & \qw & \qw 
			& \qw & \qw & \qw \\
		& \qw & \qw & \ctrl{1} & \qw & \ghost{D} & \qw & \qw & \ctrl{1} & \qw 
			& \qw & \qw \\
		& \gate{H} & \gate{T} & \targ & \gate{T^\dagger} & \ghost{D} & \targ\qwx[-2] 
			& \gate{T} & \targ & \gate{T^\dagger} & \gate{H} & \qw
	}
}
& \hspace{-1.2em} \tag{\showeqno}\label{circ:fk}
\end{array}
\]
The function $f_k$ is defined by the recurrence
\begin{align*}
	f_0(x) &= 0 \\  f_1(x) &= x_1 \\
	f_k(x) &= x_k\cdot f_{k-1}(x) + f_{k-2}(x)
\end{align*}
The contribution of $f_{k-2}$ is due to the relative phase of 
$(-1)^{y\cdot f_{k-1}}$
from the unmatched multiplication, which eventually gets swapped 
\emph{back} into the state. For instance, for $k=4$ we have
\[
	f_4 (x_1,x_2,x_3,x_4) = x_1x_2x_3x_4 \oplus x_1x_4 \oplus x_3x_4.
\]
Different recurrences and initial conditions can be obtained by tuning the
construction with additional Clifford gates, or by switching to matched 
multiplication. In particular, the relative phase $4$ qubit Toffoli in 
\cref{fig:maslov} is obtained by using matched multiplication for 
$f_2$ at no additional $T$-cost. The result is the recurrence
\begin{align*}
	f_0(x) &= 0 \qquad  &&f_1(x)= x_1 \\ f_2(x) &= x_1x_2 \qquad 
		&&f_3(x) = x_1x_2x_3  \\
	f_k(x) &= x_k\cdot f_{k-1} + f_{k-2}
\end{align*}
\end{construction}

\begin{proposition}
	There exists a maximal degree Boolean function $f:\Z_2^k \to
        \Z_2$ such that, without ancillas, the oracle $U_f$ can be
        implemented up to a phase in the controls and the target by a
        ciruit of $T$-count $4(k-1)$.
\end{proposition}

\begin{proposition}
	There exists a maximal degree Boolean function $f:\Z_2^k \to
        \Z_2$ such that, with a single dirty ancilla, the oracle $U_f$
        can be implemented by a ciruit of $T$-count $8(k-1)$.
\end{proposition}

\begin{remark}
The construction in \cref{ex:efficient} is notable in that matches
or outperforms the best-known \cite{mscrd19} $T$-count
for any degree $k$ function, \emph{without ancillas, measurement, or
classical control} but at the expense of a relative phase. Specifically, the
above construction uses $4(\deg(f_k) - 1)$ $T$ gates, where 
$\deg(f_k) -1 \leq c_{\land}(f_k)$, the multiplicative complexity of $f_k$. 

While functions derivable with this construction are not
likely to be of practical use for circuit designers, they may be useful in
automated circuit synthesis such as LUT-based logic synthesis \cite{srwm19},
where arbitrary Boolean functions on a small number of bits are used to
synthesize larger oracles. For instance, $f_4$ and $f_5$ --- corresponding
to the spectral classes \texttt{\#0888} and \texttt{\#a8808000} 
\cite{msrd20-2}, respectively --- reduce the best-known, space-minimal
constructions from $T$-count $77$ and $490$ to $12$ and $16$ up to phase,
or $24$ and $32$ exactly \cite{msrd20-2}. We leave it as an area of future
work to identify more distinct spectral classes efficiently implementable using
variations of this construction.
\end{remark}

We end the section by giving novel relative phase implementations of the
$k$-control Toffoli gate, our best construction of which halves the $T$-count 
of the best-known ancilla-free circuit. To simplify our presentation, we 
introduce shorthand for two types of relative phase gates: $U^\bullet$ and 
$U^\star$, corresponding to whether the relative phase is on the controls
and ancillas, or controls, ancillas, and target, respectively. We use boxes
on dirty ancillas to denote relative phases. Hence,
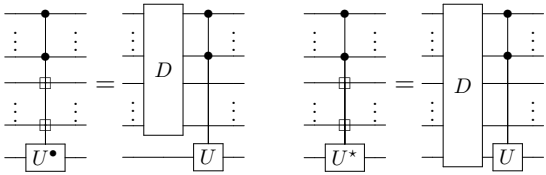
\[
	\scalebox{0.8}{
	\Qcircuit @C=.5em @R=.9em {
		& \qw & \ctrl{6} & \qw & \qw \\
		& \vdots & & \vdots & \\
		& \qw & \ctrl{4} & \qw & \qw \\
		\push{\rule{0em}{.4em}}& \qw & \ctrlb{3} & \qw & \qw \\
		& \vdots & & \vdots & \\
		& \qw & \ctrlb{1} & \qw & \qw \\
		& \qw & \gate{U^\bullet} & \qw & \qw
	}
	}
	~\mp{8}{=}
	\scalebox{0.8}{
	\Qcircuit @C=.5em @R=.4em {
		& \qw & \multigate{5}{D} &  \ctrl{6} & \qw & \qw \\
		\push{\rule{0em}{.3em}}& \vdots & & & \vdots & \\
		& \qw & \ghost{D} & \ctrl{4} & \qw & \qw \\
		& \qw & \ghost{D} & \qw & \qw & \qw \\
		\push{\rule{0em}{.3em}}& \vdots & & & \vdots & \\
		& \qw & \ghost{D} & \qw & \qw & \qw \\
		& \qw & \qw & \gate{U} & \qw & \qw
	}
	}
	\qquad
	\scalebox{0.8}{
	\Qcircuit @C=.5em @R=.9em {
		& \qw & \ctrl{6} & \qw & \qw \\
		& \vdots & & \vdots & \\
		& \qw & \ctrl{4} & \qw & \qw \\
		\push{\rule{0em}{.4em}}& \qw & \ctrlb{3} & \qw & \qw \\
		& \vdots & & \vdots & \\
		& \qw & \ctrlb{1} & \qw & \qw \\
		& \qw & \gate{U^\star} & \qw & \qw
	}
	}
	~\mp{8}{=}
	\scalebox{0.8}{
	\Qcircuit @C=.5em @R=.4em {
		& \qw & \multigate{6}{D} &  \ctrl{6} & \qw & \qw \\
		\push{\rule{0em}{.3em}}& \vdots & & & \vdots & \\
		& \qw & \ghost{D} & \ctrl{4} & \qw & \qw \\
		& \qw & \ghost{D} & \qw & \qw & \qw \\
		\push{\rule{0em}{.3em}}& \vdots & & & \vdots & \\
		& \qw & \ghost{D} & \qw & \qw & \qw \\
		& \qw & \ghost{D} & \gate{U} & \qw & \qw
	}
	}
\]
where the gates $D$ are some unspecified diagonal gates. 
We denote the inverse of
$U^\bullet$ or $U^\star$ by ${}^\bullet U$ or ${}^\star U$, 
respectively.

\begin{construction}[Ancilla-free Toffoli gates]\label{ex:toffolis}
We first note that we can use the relative phase Toffoli of 
\cref{ex:cxbulletdirty} together with matched multiplication, which cancels
each of the relative phases, to get an improved (in the $T$-count) 
implementation of the $\Lambda_k(iX)$ gate ($T$-count $16(k-3) + 4$ when $k \geq 4$):
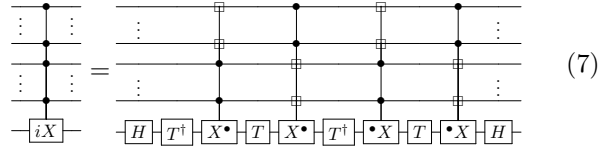
\[
\begin{array}{lr}
\hspace{-1em}
\scalebox{0.68}{
	\Qcircuit @C=.5em @R=0.9em {
		& \qw & \ctrl{6} & \qw & \qw \\
		 & & & & \\
		& \ustick{\vdots}\qw & \ctrl{4} 
			& \ustick{\vdots}\qw & \qw \\
		& \qw & \ctrl{3} & \qw & \qw \\
		 & & & & \\
	 	& \ustick{\vdots}\qw & \ctrl{1} 
			& \ustick{\vdots}\qw & \qw \\
		& \qw & \gate{iX} & \qw & \qw \\
	}
}
~\mp{7.3}{=}
\scalebox{0.68}{
	\Qcircuit @C=.5em @R=0.9em {
		& \qw & \qw & \ctrlb{3} & \qw & \ctrl{6} & \qw & \ctrlb{3} & \qw 
			& \ctrl{6} & \qw & \qw \\
		& \vdots & & & & & & & & & \vdots & \\
		& \qw & \qw & \ctrlb{1} & \qw & \ctrl{4} & \qw & \ctrlb{1} & \qw 
			& \ctrl{4} &\qw & \qw \\
		& \qw & \qw & \ctrl{3} & \qw & \ctrlb{3} & \qw & \ctrl{3} 
			& \qw & \ctrlb{3} & \qw & \qw \\
		& \vdots & & & & & & & & & \vdots & \\
		& \qw & \qw & \ctrl{1} & \qw & \ctrlb{1} & \qw & \ctrl{1} 
			& \qw & \ctrlb{1} & \qw & \qw \\
		& \gate{H} & \gate{T^\dagger} & \gate{X^\bullet} 
			& \gate{T} & \gate{X^\bullet} & \gate{T^\dagger} 
			& \gate{{}^\bullet X} & \gate{T} & \gate{{}^\bullet X} & \gate{H} & \qw
	}
}
& \hspace{-1em}\tag{\showeqno}\label{circ:cix}
\end{array}
\]

Next we leverage the un-matched multiplication, placing all but one control 
on the un-matched Toffoli gate and using the single dirty ancilla relative phase
Toffoli from \cref{ex:cxbulletdirty} to perform it up to phase. The result is an 
ancilla-free $k$-controlled Toffoli gate using $8(k-2)$ $T$ gates, roughly half 
that of the best-known ancilla-free $k$-controlled 
Toffoli gate, at the expense of a target-dependent phase:
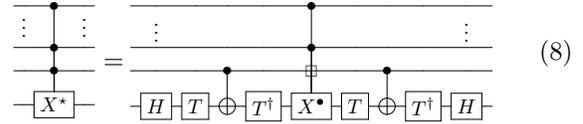
\[
	\begin{array}{lr}
	\scalebox{0.77}{
	\Qcircuit @C=.5em @R=0.9em {
		& \qw & \ctrl{4} & \qw & \qw \\
		& & & & \\
		& \ustick{\vdots}\qw & \ctrl{2} & \ustick{\vdots}\qw & \qw \\
		& \qw & \ctrl{1} & \qw & \qw \\
		& \qw & \gate{X^\star} & \qw & \qw
	}
	}
	~\mp{6.5}{=}
	\scalebox{0.77}{
	\Qcircuit @C=.5em @R=0.9em {
		& \qw & \qw & \qw & \qw & \ctrl{4} & \qw & \qw 
			& \qw & \qw & \qw \\
		& \vdots & & & & & & & & \vdots & \\
		& \qw & \qw & \qw & \qw & \ctrl{2} & \qw & \qw 
			& \qw & \qw & \qw \\
		& \qw & \qw & \ctrl{1} & \qw & \ctrlb{1} & \qw & \ctrl{1} 
			& \qw & \qw & \qw \\
		& \gate{H} & \gate{T} & \targ & \gate{T^\dagger} 
			& \gate{X^\bullet} & \gate{T} & \targ & \gate{T^\dagger} 
			& \gate{H} & \qw
	}
	}
	& \tag{\showeqno}\label{circ:cxstar}
	\end{array}
\]

Our final ancilla-free construction, below, uses the previous 
circuit to perform a $k$-controlled Toffoli up to a phase 
\emph{only on the controls}. The construction uses matched multiplication 
to eliminate the target-dependent phases produced by the intermediate 
Toffoli gates. The result is an additional $16$ $T$ gates of savings compared 
to the $k$-controlled $iX$ gate above ($T$-count $16(k-4) + 4$ when $k \geq 5$):
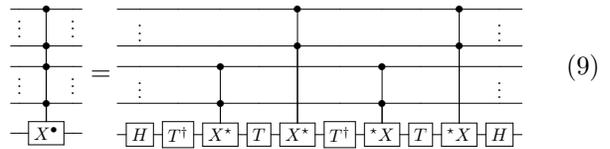
\[
\begin{array}{cc}
\hspace{-1em}
\scalebox{0.68}{
	\Qcircuit @C=.5em @R=0.9em {
		& \qw & \ctrl{6} & \qw & \qw \\
		& & & & \\
		& \ustick{\vdots}\qw & \ctrl{4} & \ustick{\vdots}\qw & \qw \\
		& \qw & \ctrl{3} & \qw & \qw \\
		& & & \\
		& \ustick{\vdots}\qw & \ctrl{1} & \ustick{\vdots}\qw & \qw \\
		& \qw & \gate{X^\bullet} & \qw & \qw \\
	}
}
~\mp{7.3}{=}
\scalebox{0.68}{
	\Qcircuit @C=.5em @R=0.9em {
		& \qw & \qw & \qw & \qw & \ctrl{6} & \qw & \qw & \qw & \ctrl{6} 
			& \qw & \qw \\
		& \vdots & & & & & & & & & \vdots & \\
		& \qw & \qw & \qw & \qw & \ctrl{4} & \qw 
			& \qw & \qw & \ctrl{4} & \qw & \qw \\
		& \qw & \qw & \ctrl{3} & \qw & \qw & \qw & \ctrl{3} & \qw & \qw 
			& \qw & \qw \\
		& \vdots & & & & & & & & & \vdots & \\
		& \qw & \qw & \ctrl{1} & \qw & \qw 
			& \qw & \ctrl{1} & \qw & \qw & \qw & \qw \\
		& \gate{H} & \gate{T^\dagger} & \gate{X^\star} & \gate{T} & \gate{X^\star} 
			& \gate{T^\dagger} & \gate{{}^\star X} & \gate{T} 
			& \gate{{}^\star X} & \gate{H} & \qw
	}
}
& \hspace{-1em}\tag{\showeqno}\label{circ:cxbullet}
\end{array}
\]
\end{construction}

\begin{proposition}
	Let $k\in\Z^{\geq 5}$. Without ancillas, the $\Lambda_k(X)$
        gate can be implemented up to a phase in the controls by a
        circuit of $T$-count $16(k-4) + 4$.
\end{proposition}
        
\begin{proposition}
	Let $k\in\Z^{\geq 3}$. Without ancillas, the $\Lambda_k(X)$
        gate can be implemented up to a phase in the controls and the
        target by a circuit of $T$-count $8(k-2)$.
\end{proposition}        

\begin{remark}
Combining the $X^\star$ construction with the Gidney logical-AND \cite{g18}
gives a method of further reducing the $T$-count of the multiply-controlled
Toffoli gate (up to phase) when \emph{some} ancillas are available. 
In particular, by using Gidney's logical-AND to initialize and terminate
$m$ temporary products with $m$ clean ancillas and $4m$ $T$ gates, this
gives a $T$-count of $4m + 8(k-m-2)$ for a $k$-controlled $X^\star$
with $m$ clean ancillas, where $m\leq k-1$.
\end{remark}

\section{Measurement-assisted uncomputation}

The last technique that we study generalizes the constructions of
Gidney and Jones for terminating a temporary product
\cite{j13,g18}. Recall that by termination we mean
the dual of initialization, which is distinguished from the (unitary)
process of uncomputation.

To terminate an ancilla in the temporary state $\ket{f(x)}$, one
typically uncomputes $f$. Instead, we can swap the state into the
phase space by applying a Hadamard gate:
\[
	H\ket{f(x)} = \frac{1}{\sqrt{2}} \sum_{y\in\Z_2}(-1)^{yf(x)}\ket{y}.
\]
Measuring the ancillary qubit then leaves a phase of $1$ or a phase of
$(-1)^{f(x)}$. In the latter case, this phase can be corrected by a
classically-controlled $(-1)^{f(x)}$ phase oracle. This is reflected
in the following sequence of circuit equalities:
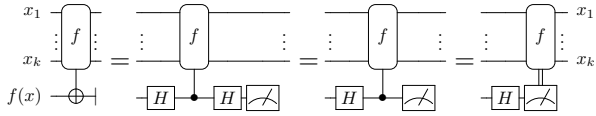
\[	\hspace*{5pt}
	\scalebox{0.68}{
	\Qcircuit @C=.3em @R=0.9em {
		\lstick{x_1} & \qw & \multimeasure{2}{f} & \qw & \qw \\
		& \vdots & & \vdots & \\
		\lstick{x_k} & \qw & \ghost{f} & \qw & \qw \\
		\lstick{f(x)} & \qw & \targ\qwx[-1] & {|}\qw & \push{\rule{0em}{1.2em}}
	}
	}
	~\mp{5.7}{=}
	\scalebox{0.68}{
	\Qcircuit @C=.3em @R=0.9em {
		& \qw & \qw & \multimeasure{2}{f} & \qw & \qw & \qw & \qw \\
		& \vdots & & & & & \vdots & \\
		& \qw & \qw & \ghost{f} & \qw & \qw & \qw & \qw \\
		& \qw & \gate{H} & \ctrl{-1} & \gate{H} & \meter & 
	}
	}
	~\mp{5.7}{=}
	\scalebox{0.68}{
	\Qcircuit @C=.3em @R=0.9em {
		& \qw & \qw & \multimeasure{2}{f} & \qw & \qw & \qw \\
		& \vdots & & & & \vdots & \\
		& \qw & \qw & \ghost{f} & \qw & \qw & \qw \\
		& \qw & \gate{H} & \ctrl{-1} & \meter & 
	}
	}
	~\mp{5.7}{=}
	\scalebox{0.68}{
	\Qcircuit @C=.3em @R=0.9em {
		& \qw & \qw & \multimeasure{2}{f} & \qw & \rstick{x_1}\qw \\
		& \vdots & & & \vdots & \\
		& \qw & \qw & \ghost{f} & \qw & \rstick{x_k}\qw \\
		& \qw & \gate{H} & \meter\cwx[-1] & 
	}
	}
\]
Note that the second equality follows from the fact that single qubit
gates preceding a discarded measurement can be dropped \cite{ku17}.
We once again summarize this fact below.

\begin{quote}
	\textit{
	A temporary value $\ket{f(x)}$ can be terminated by measuring
        in the $X$-basis and applying a classically-controlled
        $(-1)^{f(x)}$ correction.
	}
\end{quote}

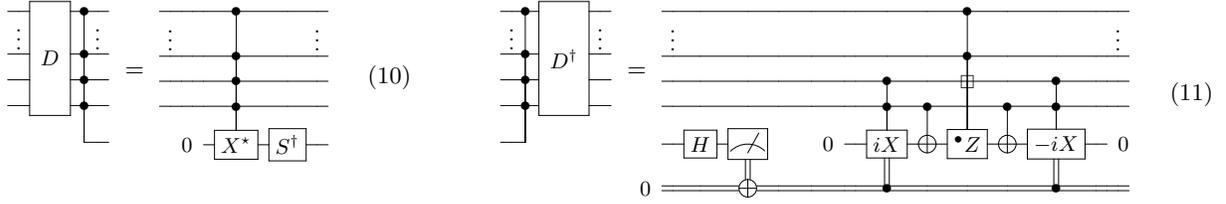
\begin{figure*}
	\begin{minipage}[t]{0.35\textwidth}
	\[
	\begin{array}{cc}
	\scalebox{0.9}{
	\Qcircuit @C=.5em @R=0.27em {
		& \qw & \multigate{4}{D} & \ctrl{5} & \qw & \qw  \\
		\push{\rule{0em}{.5em}} & \vdots & & & \vdots & \\
		& \qw & \ghost{D} & \ctrl{3} & \qw & \qw  \\
		& \qw & \ghost{D} & \ctrl{2} & \qw & \qw  \\
		& \qw & \ghost{D} & \ctrl{1} & \qw & \qw  \\
		\push{\rule{0em}{1.9em}} & & & {} & \qw & \qw  \\
	}
	}
	~~\mp{7.3}{=}~
	\scalebox{0.9}{
		\Qcircuit @C=.5em @R=0.9em {
		& \qw & \qw & \qw & \qw & \ctrl{5} & \qw & \qw & \qw  \\
		& \vdots & & & & & & \vdots & \\
		& \qw & \qw & \qw & \qw & \ctrl{3} & \qw & \qw & \qw  \\
		& \qw & \qw & \qw & \qw & \ctrl{2} & \qw & \qw & \qw  \\
		& \qw & \qw& \qw & \qw & \ctrl{1} & \qw & \qw & \qw  \\
		& & & &  \lstick{0} & \gate{X^\star} & \gate{S^\dagger} & \qw & \qw  \\
	}
	}
	& \tag{\showeqno}\label{circ:andk}
	\end{array}
	\]
	\end{minipage}
	\qquad
	\begin{minipage}[t]{0.60\textwidth}
	\[
	\begin{array}{cc}
	\scalebox{0.9}{
	\Qcircuit @C=.5em @R=0.27em {
		& \qw & \ctrl{5} & \multigate{4}{D^\dagger} & \qw & \qw  \\
		\push{\rule{0em}{.5em}}& \vdots & & & \vdots & \\
		& \qw & \ctrl{3} & \ghost{D^\dagger} & \qw & \qw  \\
		& \qw & \ctrl{2} & \ghost{D^\dagger} & \qw & \qw  \\
		& \qw & \ctrl{1} &  \ghost{D^\dagger} & \qw & \qw  \\
		\push{\rule{0em}{1.9em}} & \qw & {}\qw  & & &  \\
		& & & &
	}
	}
	~~\mp{7.3}{=}~
	\scalebox{0.9}{
		\Qcircuit @C=.5em @R=0.9em {
		& \qw & \qw & \qw & \qw & \qw & \qw & \qw & \qw & \qw & \ctrl{5} 
			& \qw & \qw & \qw & \qw & \qw & \qw  \\
		& \vdots & & & & & & & & & & & & & & \vdots & \\
		& \qw & \qw & \qw & \qw & \qw & \qw & \qw & \qw & \qw & \ctrl{3} 
			& \qw & \qw & \qw & \qw & \qw & \qw  \\
		& \qw & \qw & \qw & \qw & \qw & \qw & \qw & \ctrl{2} & \qw & \ctrlb{2} 
			& \qw & \ctrl{2} & \qw & \qw & \qw & \qw  \\
		& \qw & \qw & \qw & \qw & \qw & \qw & \qw & \ctrl{1} & \ctrl{1} & \qw 
			& \ctrl{1} & \ctrl{1} & \qw & \qw & \qw & \qw  \\
		& \qw & \gate{H} & \meter\cwx[1]  & & \push{\rule{2em}{0em}} & \lstick{0} & \qw & \gate{iX} 
			& \targ & \gate{{}^\bullet Z} & \targ & \gate{-iX} & \qw & \rstick{0}\qw & \\
		\lstick{0} & \cw & \cw & \ctarg & \cw & \cw & \cw & \cw & \cctrl{-1} & \cw 
			& \cw & \cw & \cctrl{-1} & \cw & \cw & \cw & \cw
	}
	}\hspace{-1em}
	& \tag{\showeqno}\label{circ:unkand}
	\end{array}
	\]
	\end{minipage}
\caption{Initializing and terminating a temporary logical AND of $k$ bits. Combined $T$-count $16(k - 2) - 10\pm 2$.}
\label{fig:tla}
\end{figure*}

In the case of the $2$-qubit Toffoli gate, the classically-controlled
correction of $(-1)^{x_1x_2}$ can be implemented using only Clifford
gates. In contrast, in the general case, correcting the phase
$(-1)^{f(x)}$ might require $T$ gates. This can nonetheless still reduce
the $T$-count, since uncomputing $f(x)$ without
measurement requires computing the phase
\[
	\ket{x}\ket{y} \mapsto (-1)^{yf(x)}\ket{x}\ket{y}
\]
conrresponding to an extra quantum control.

\begin{construction}[Terminating $\Lambda_k(X)$]
	Consider the logical product of $k$ bits $\ket{(x_1\cdots x_k)}$, 
	which can be initialized with
	a clean ancilla by applying a multiply-controlled $iX$ gate and 
	an $S^\dagger$:
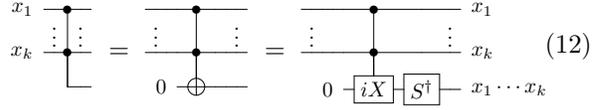
	\[\hspace*{-1em}
	\begin{array}{cc}
	\scalebox{0.8}{
	\Qcircuit @C=.5em @R=0.9em {
		\lstick{x_1} & \qw & \ctrl{3} & \qw & \qw  \\
		& \vdots & & \vdots & \\
		\lstick{x_k} & \qw & \ctrl{1} & \qw & \qw  \\
		\push{\rule{0em}{1.2em}} & & {} & \qw & \qw  \\
	}
	}	
	~~\mp{5}{=}~
	\scalebox{0.8}{
		\Qcircuit @C=.5em @R=0.9em {
		& \qw & \qw & \qw & \ctrl{3} & \qw & \qw & \qw & \qw  \\
		& \vdots & & & & & & \vdots & \\
		& \qw & \qw & \qw & \ctrl{1} & \qw & \qw & \qw & \qw  \\
		\push{\rule{0em}{1.2em}} & & & \lstick{0} & \targ  & \qw & \qw & \qw & \qw  \\
	}
	}
	~~\mp{5}{=}~
	\scalebox{0.8}{
		\Qcircuit @C=.5em @R=0.9em {
		& \qw & \qw & \qw & \qw & \ctrl{3} & \qw & \qw & \rstick{x_1}\qw  \\
		& \vdots & & & & & & \vdots & \\
		& \qw & \qw& \qw & \qw & \ctrl{1} & \qw & \qw & \rstick{x_k}\qw  \\
		& & & &  \lstick{0} & \gate{iX} & \gate{S^\dagger} & \qw & \rstick{x_1\cdots x_k}\qw  \\
	}
	}
	& \tag{\showeqno}\label{circ:kand}
	\end{array}
	\]
	Using the $\Lambda_k(iX)$ gate from \cref{sec:ancillafree}, the logical
	product above uses $16(k-3) + 4$ $T$ gates.

	After this temporary product is no longer needed, we can
        terminate it by measuring in the $X$ basis, resulting in the
        state $\ket{(x_1\cdots x_k)}$ or $(-1)^{x_1\cdots
          x_k}\ket{(x_1\cdots x_k)}$. Correcting this phase requires a
        $\Lambda_{k-1}(Z)$ gate. Using the methods from
        \cref{sec:ancilla}, this gate can be implemented with a
        single ancilla in $T$-count $16(k-3)$. With
        non-destructive measurements, this can be further reduced to
        $16(k-4)+4$ by re-using the measured qubit as a clean ancilla to 
	initialize a temporary product of two bits:
	\[\hspace*{1em}
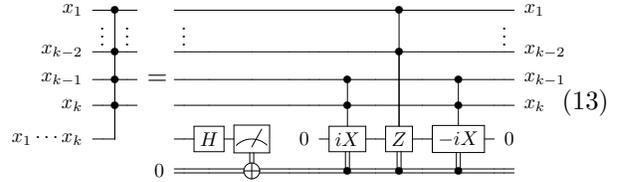
	\begin{array}{c}
	\scalebox{0.75}{
	\Qcircuit @C=.5em @R=0.2em {
		\lstick{x_1} & \qw & \ctrl{5} & \qw & \qw & \push{\rule{0em}{1.2em}}  \\
		& \vdots & & \vdots & & \push{\rule{0em}{.5em}} \\
		\lstick{x_{k-2}} & \qw & \ctrl{3} & \qw & \qw & \push{\rule{0em}{1.2em}}  \\
		\lstick{x_{k-1}} & \qw & \ctrl{2} & \qw & \qw & \push{\rule{0em}{1.2em}}  \\
		\lstick{x_k} & \qw & \ctrl{1} & \qw & \qw & \push{\rule{0em}{1em}}  \\
		\lstick{x_1 \cdots x_k} & \qw & {}\qw  & \push{\rule{0em}{2em}} & &  \\
		& & & & \push{\rule{0em}{1.2em}} \\
	}
	}
	\mp{7.7}{=}
	\scalebox{0.75}{
		\Qcircuit @C=.5em @R=0.2em {
		\push{\rule{0em}{1.2em}} & \qw & \qw & \qw & \qw & \qw & \qw & \qw & \qw & \ctrl{5} 
			& \qw & \qw & \qw & \qw & \rstick{x_1}\qw  \\
		\push{\rule{0em}{.5em}} & \vdots & & & & & & & & & & & & \vdots & \\
		\push{\rule{0em}{1.2em}} & \qw & \qw & \qw & \qw & \qw & \qw & \qw & \qw & \ctrl{3} 
			& \qw & \qw & \qw & \qw & \rstick{x_{k-2}}\qw  \\
		\push{\rule{0em}{1.2em}} & \qw & \qw & \qw & \qw & \qw & \qw & \ctrl{1} & \qw & \qw 
			& \qw & \ctrl{1} & \qw & \qw & \rstick{x_{k-1}}\qw  \\
		\push{\rule{0em}{1em}} & \qw & \qw & \qw & \qw & \qw & \qw & \ctrl{1} & \qw & \qw 
			& \qw & \ctrl{1} & \qw & \qw & \rstick{x_k}\qw  \\
		\push{\rule{0em}{2em}} & \qw & \gate{H} & \meter\cwx[1] & \push{\rule{1em}{0em}} & & \lstick{0} & \gate{iX} & \qw 
			& \gate{Z} & \qw & \gate{-iX} & \rstick{0}\qw &  \\
		 \lstick{0}\push{\rule{0em}{1.2em}} & \cw & \cw & \ctarg & \cw & \cw & \cw &  \cctrl{-1} & \cw & \cctrl{-1} 
			& \cw &  \cctrl{-1} & \cw & \cw & \cw  \\
	}
	} \\
	\tag{\showeqno}\label{circ:kunand}
	\end{array}
	\]
\end{construction}

In the case where the oracle is implemented up to phase, the $T$-count of the 
final phase correction can sometimes be reduced further by applying the phase 
correction \emph{itself up to phase}. The following construction gives a 
measurement-assisted termination circuit for a temporary logical $3$-AND based
on Maslov's Toffoli $4$.

\begin{construction}[Terminating Maslov's $\Lambda_3(X^\star)$]
	The ternary logical AND $f(x)=x_1x_2x_3$ can be initialized
        with a clean ancilla using Maslov's $4$-qubit Toffoli gate
        (see \cref{fig:maslov}) up to a phase of $i^{x_1x_2}$. In
        particular,
	\[
	\scalebox{0.8}
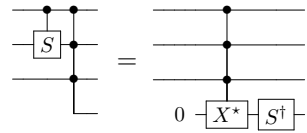
{
	\Qcircuit @C=.5em @R=0.35em @!R {
		& \qw & \ctrl{1} & \ctrl{3} & \qw & \qw  \\
		& \qw & \gate{S} & \ctrl{2} & \qw & \qw  \\
		& \qw & \qw & \ctrl{1} & \qw & \qw  \\
		& & & {} & \qw & \qw &  \\
	}
	}
	~\mp{6}{=}~
	\scalebox{0.8}{
		\Qcircuit @C=.5em @R=0.2em @!R {
		& \qw & \qw & \qw & \qw & \ctrl{3} & \qw & \qw  \\
		& \qw & \qw & \qw & \qw & \ctrl{2} & \qw & \qw  \\
		& \qw & \qw& \qw & \qw & \ctrl{1} & \qw & \qw  \\
		& & & &  \lstick{0} & \gate{X^\star} & \gate{S^\dagger} & \qw &  \\
	}
	}
	\]
	where the $3$-control $X^\star$ gate is the right hand side of
	\cref{fig:maslov}. In this case, measuring the product in the $X$ basis
	gives either a phase of $i^{x_1x_2}$ if the result is $0$, or a phase of
	$i^{x_1x_2}(-1)^{x_1x_2x_3}$ if the measurement result is $1$. In the
	former case, the phase of $i^{x_1x_2}$ can be corrected with a controlled
	$S^\dagger$ gate in $3$ $T$ gates, while the latter phase of 
	$i^{x_1x_2}(-1)^{x_1x_2x_3}$ can be corrected with a $\Lambda_2(-iZ)$
	gate using $4$ $T$ gates. The corresponding circuit is shown below.
	\[
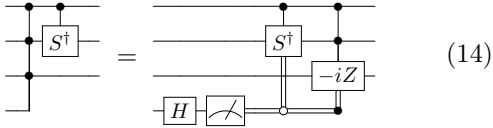
	\begin{array}{cc}
	\scalebox{0.8}{
	\Qcircuit @C=.5em @R=0.2em @!R {
		& \qw & \ctrl{3} & \ctrl{1} & \qw & \qw  \\
		& \qw & \ctrl{2} & \gate{S^\dagger} & \qw & \qw  \\
		& \qw & \ctrl{1} & \qw & \qw & \qw  \\
		& \qw & {}\qw  & & &  \\
	}
	}
	~~\mp{6}{=}~
	\scalebox{0.8}{
		\Qcircuit @C=.5em @R=0.2em @!R {
		& \qw & \qw & \qw & \ctrl{1} & \ctrl{2} & \qw  \\
		& \qw & \qw & \qw & \gate{S^\dagger} & \ctrl{1} & \qw  \\
		& \qw & \qw & \qw & \qw & \gate{-iZ} & \qw  \\
		& \gate{H} & \meter  & \cw & \cctrlo{-2} & \cctrl{-1} &  \\
	}
	}
	& \tag{\showeqno}\label{circ:3unand}
	\end{array}
	\]
\end{construction}

We close by giving an efficient logical $k$-AND using our $\Lambda_k(X^\star)$
to initialize the product, and measurement to terminate it (\cref{fig:tla}).
The termination construction shaves roughly $16$ $T$ gates off the cost
to compute the product.

\begin{proposition}
\label{prop:logicalandmeas}
A logical AND of $k$ bits can be initialized up to relative phase with $8(k-2)$
$T$ gates and terminated with either $8(k-4)$ or $8(k-4) + 4$ $T$ gates.
\end{proposition}

As a corollary we additionally 
obtain a Jones-style circuit for the $\Lambda_k(X)$ gate which 
uses a single clean ancilla, measurements and classical control.

\begin{proposition}
	Let $k\in\Z^{\geq 4}$. With a single dirty ancilla and measurements, the
        $\Lambda_k(X)$ gate can be implemented by a circuit of
        $T$-count $16(k-3)+4$.
\end{proposition}

\section{Conclusion}

In this paper we described a number of techniques which use the
phase/state duality in order to efficiently implement quantum oracles,
most notably in the low-space regime. These techniques generalize,
among others, the relative phase Toffolis of Maslov and Selinger, as
well as the measurement-assisted uncomputations of Gidney and
Jones. Using these techniques, we developed several new circuit
constructions. These constructions, which are summarized in
\cref{tab:foo}, include circuits for Toffoli gates, multiplying
Boolean functions, and high-degree classical gates.

\begin{acknowledgments}
The authors would like to thank Craig Gidney
for helpful and enlightening comments on an earlier version of this paper.
\end{acknowledgments}

%

\appendix
\onecolumngrid

\section{Implications of relative phases}\label{app:rphase}

As most of our constructions have some form of relative phase associated with
them, care must be taken when using them. Here we show that any permutation
or reversible circuit which is later uncomputed can be replaced with a relative
phase implementation, provided the interior computation doesn't modify the
basis state of any qubits on which there is a relative phase.

\begin{proposition}
Let $U_f$ be an oracle for some Boolean function 
$f:\Z_2^n\rightarrow \Z_2$
and let $U$ be some unitary transformation on $m>n$ qubits.
Without loss of generality, 
if $U$ is constant up to phase on the first $n$ qubits, then
\[
	(U_f^\dagger \otimes I) U (U_f \otimes I) = (\widetilde{U_f}^\dagger \otimes I) U (\widetilde{U_f} \otimes I)
\]
for any relative phase implementation $\widetilde{U_f}$ of $f$.
\end{proposition}
\begin{proof}
First observe that if $U$ is constant up to phase in the first $n$ qubits, 
then $U$ is equivalent to some circuit where the first $n$ qubits are \emph{only
used as controls}. In particular, $U$ can be written as a product of a diagonal unitary $V$
and $2^n$ matrices controlled on the $2^n$ basis states of the first $n$ qubits:
\[
	U = \prod_{i=1}^{2^n}\Lambda_{n}^{i}(U_i)
\]
where we use $\Lambda_{n}^{i}(U_i)$ to denote the application of $U_i$ controlled on
the first $n$ qubits having basis state equal to the binary expansion of $i$.

Now recalling that a generalized permutation $\widetilde{U_f}$ may be factored
equally as $DU_f = \widetilde{U_f} = U_fD'$ for some diagonal matrices $D,D'$, we
can proceed by calculation.
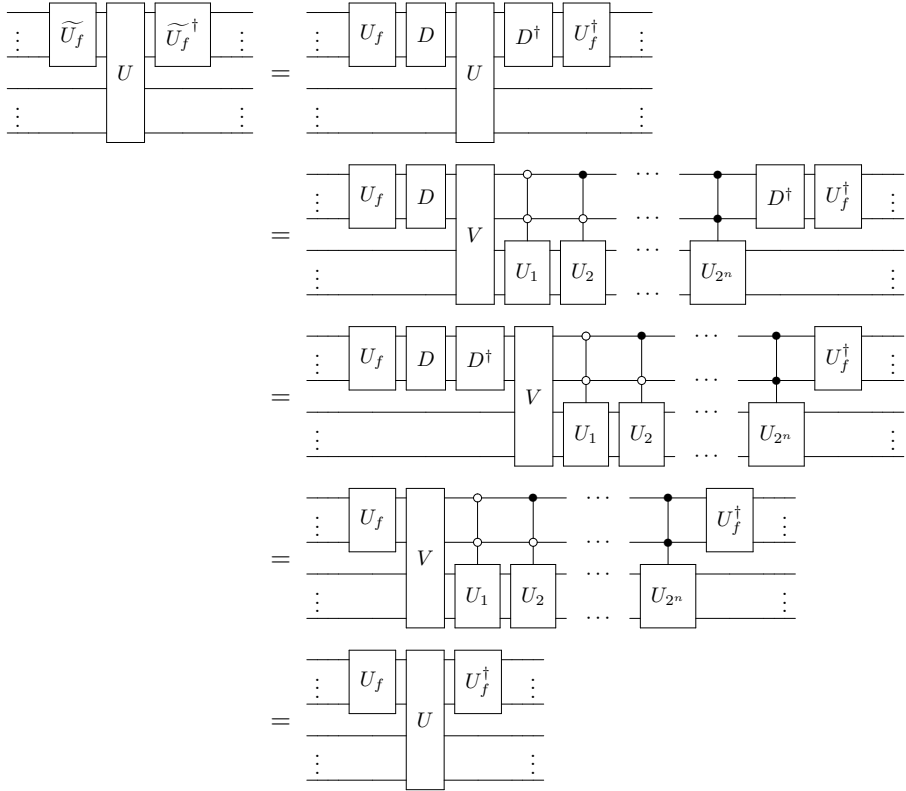
\begin{align*}
\scalebox{0.8}{
	\Qcircuit @C=.5em @R=0.6em {
	&\qw &\qw &\qw &\multigate{2}{\widetilde{U_f}} &\multigate{5}{U} &\multigate{2}{\widetilde{U_f}^\dagger} 
		&\qw &\qw &\qw & \qw \\
	& \vdots & & & & & & & & \vdots \\
	&\qw &\qw &\qw &\ghost{\widetilde{U_f}} & \ghost{U} &\ghost{\widetilde{U_f}^\dagger} &\qw &\qw &\qw & \qw \\
        & \qw & \qw &\qw & \qw & \ghost{U} & \qw & \qw & \qw & \qw & \qw \\
	& \vdots & & & & & & & & \vdots \\
        & \qw & \qw &\qw & \qw & \ghost{U} & \qw & \qw & \qw & \qw & \qw
}
}
~~&\mp{7}{=}~
\scalebox{0.8}{
	\Qcircuit @C=.5em @R=0.6em {
	&\qw &\qw &\qw & \multigate{2}{U_f} & \multigate{2}{D} &\multigate{5}{U} 
		& \multigate{2}{D^\dagger} &\multigate{2}{U_f^\dagger} &\qw &\qw &\qw & \qw \\
	& \vdots & & & & & & & & & & \vdots \\
	&\qw &\qw &\qw &\ghost{U_f} &\ghost{D} & \ghost{U} &\ghost{D^\dagger} &\ghost{U_f^\dagger} &\qw &\qw &\qw & \qw \\
        & \qw & \qw &\qw & \qw & \qw & \ghost{U} & \qw & \qw & \qw & \qw & \qw & \qw \\
	& \vdots & & & & & & & & & & \vdots \\
        & \qw & \qw &\qw & \qw & \qw & \ghost{U} & \qw & \qw & \qw & \qw & \qw & \qw
}
} \\
~~&\mp{7}{=}~
\scalebox{0.8}{
	\Qcircuit @C=.5em @R=0.6em {
	&\qw &\qw &\qw & \multigate{2}{U_f} & \multigate{2}{D} & \multigate{5}{V} 
		& \ctrlo{2} & \ctrl{2} & \qw & \push{\rule{.5em}{0em}} 
		& \cdots & \push{\rule{.5em}{0em}} & & \ctrl{2} & \multigate{2}{D^\dagger} 
		&\multigate{2}{U_f^\dagger} &\qw &\qw &\qw & \qw \\
	& \vdots & & & & & & & & & & & & & & & & & & \vdots \\
	&\qw &\qw &\qw &\ghost{U_f} & \ghost{D} & \ghost{V} & \ctrlo{1} & \ctrlo{1} & \qw & & \cdots & & & \ctrl{1} 
		&\ghost{D^\dagger} &\ghost{U_f^\dagger} &\qw &\qw &\qw & \qw \\
        & \qw & \qw &\qw & \qw & \qw & \ghost{V} & \multigate{2}{U_1} & \multigate{2}{U_2} 
		& \qw & & \cdots & & & \multigate{2}{U_{2^n}} 
		& \qw & \qw & \qw & \qw & \qw & \qw \\
	& \vdots & & & & & & & & & & & & & & & & & & \vdots \\
        & \qw & \qw &\qw & \qw & \qw & \ghost{V} & \ghost{U_1} & \ghost{U_2} & \qw & & \cdots & & & \ghost{U_{2^n}} 
		& \qw & \qw & \qw & \qw & \qw & \qw
}
} \\
~~&\mp{7}{=}~
\scalebox{0.8}{
	\Qcircuit @C=.5em @R=0.6em {
	&\qw &\qw &\qw & \multigate{2}{U_f} & \multigate{2}{D} & \multigate{2}{D^\dagger} & \multigate{5}{V} 
		& \ctrlo{2} & \ctrl{2} & \qw & \push{\rule{.5em}{0em}} 
		& \cdots & \push{\rule{.5em}{0em}} & & \ctrl{2} 
		&\multigate{2}{U_f^\dagger} &\qw &\qw &\qw & \qw \\
	& \vdots & & & & & & & & & & & & & & & & & & \vdots \\
	&\qw &\qw &\qw &\ghost{U_f} & \ghost{D} &\ghost{D^\dagger} 
		& \ghost{V} & \ctrlo{1} & \ctrlo{1} & \qw & & \cdots & & & \ctrl{1} 
		&\ghost{U_f^\dagger} &\qw &\qw &\qw & \qw \\
        & \qw & \qw &\qw & \qw & \qw & \qw & \ghost{V} & \multigate{2}{U_1} & \multigate{2}{U_2} 
		& \qw & & \cdots & & & \multigate{2}{U_{2^n}} 
		& \qw & \qw & \qw & \qw & \qw \\
	& \vdots & & & & & & & & & & & & & & & & & & \vdots \\
        & \qw & \qw &\qw & \qw & \qw & \qw & \ghost{V} & \ghost{U_1} & \ghost{U_2} & \qw & & \cdots & & & \ghost{U_{2^n}} 
		& \qw & \qw & \qw & \qw & \qw
}
} \\
~~&\mp{7}{=}~
\scalebox{0.8}{
	\Qcircuit @C=.5em @R=0.6em {
	&\qw &\qw &\qw & \multigate{2}{U_f} & \multigate{5}{V} 
		& \ctrlo{2} & \ctrl{2} & \qw & \push{\rule{.5em}{0em}} 
		& \cdots & \push{\rule{.5em}{0em}} & & \ctrl{2} 
		&\multigate{2}{U_f^\dagger} &\qw &\qw &\qw & \qw \\
	& \vdots & & & & & & & & & & & & & & & & \vdots \\
	&\qw &\qw &\qw &\ghost{U_f}
		& \ghost{V} & \ctrlo{1} & \ctrlo{1} & \qw & & \cdots & & & \ctrl{1} 
		&\ghost{U_f^\dagger} &\qw &\qw &\qw & \qw \\
        & \qw & \qw &\qw & \qw & \ghost{V} & \multigate{2}{U_1} & \multigate{2}{U_2} 
		& \qw & & \cdots & & & \multigate{2}{U_{2^n}} 
		& \qw & \qw & \qw & \qw & \qw \\
	& \vdots & & & & & & & & & & & & & & & & \vdots \\
        & \qw & \qw &\qw & \qw & \ghost{V} & \ghost{U_1} & \ghost{U_2} & \qw & & \cdots & & & \ghost{U_{2^n}} 
		& \qw & \qw & \qw & \qw & \qw
}
} \\
~~&\mp{7}{=}~
\scalebox{0.8}{
	\Qcircuit @C=.5em @R=0.6em {
	&\qw &\qw &\qw &\multigate{2}{U_f} &\multigate{5}{U} &\multigate{2}{U_f^\dagger} 
		&\qw &\qw &\qw & \qw \\
	& \vdots & & & & & & & & \vdots \\
	&\qw &\qw &\qw &\ghost{U_f} & \ghost{U} &\ghost{U_f^\dagger} &\qw &\qw &\qw & \qw \\
        & \qw & \qw &\qw & \qw & \ghost{U} & \qw & \qw & \qw & \qw & \qw \\
	& \vdots & & & & & & & & \vdots \\
        & \qw & \qw &\qw & \qw & \ghost{U} & \qw & \qw & \qw & \qw & \qw
}
}
\end{align*}
\end{proof}

A simple rule-of-thumb for when the above proposition can be applied is whenever
$U_f$ and $U_f^\dagger$ are used as a compute/uncompute pair, and
$U_f$ does not impart a relative phase on a \emph{dirty} ancilla. In such a case,
$U$ is necessarily globally constant on the qubits used in $U_f$ in order for 
$U_f^\dagger$ to correctly uncompute $U_f$.

\section{Correctness of the logical k-AND termination}\label{app:proof}

In this appendix we establish the precise form of the relative phase $D$ in the circuit constructions in \cref{fig:tla}, and prove correctness of the termination circuit. To do so, we first establish the form of the relative phase in the single dirty ancilla $X^\bullet$.

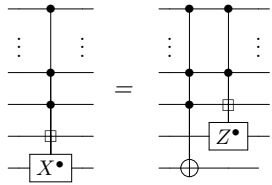
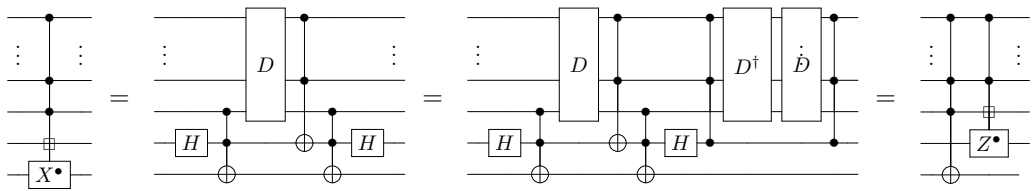
\begin{proposition}\label{prop:a1}
\begin{align*}
	\scalebox{0.8}{
	\Qcircuit @C=.5em @R=0.2em @!R {
		& \qw & \ctrl{5} & \qw & \qw \\
		& \vdots & & \vdots & \\
		\push{\rule{0em}{.5em}} & \qw & \ctrl{3} & \qw & \qw \\
		& \qw & \ctrl{2} & \qw & \qw \\
		& \qw & \ctrlb{1} & \qw & \qw \\
		& \qw & \gate{X^\bullet} & \qw & \qw
	}
	}
	~~&\mp{9}{=}~~
	\scalebox{0.8}{
	\Qcircuit @C=.5em @R=0.2em @!R {
		& \qw & \ctrl{5} & \ctrl{4} & \qw & \qw \\
		& \vdots & & & \vdots & \\
		\push{\rule{0em}{.5em}} & \qw & \ctrl{3} & \ctrl{2} & \qw & \qw \\
		& \qw & \ctrl{2} & \ctrlb{1} & \qw & \qw \\
		& \qw & \qw & \gate{Z^\bullet} & \qw & \qw \\
		& \qw & \targ & \qw & \qw
	}
	}
\end{align*}
\end{proposition}
\begin{proof}
	\begin{align*}
	\scalebox{0.8}{
	\Qcircuit @C=.5em @R=0.2em @!R {
		& \qw & \ctrl{5} & \qw & \qw \\
		& \vdots & & \vdots & \\
		\push{\rule{0em}{.5em}} & \qw & \ctrl{3} & \qw & \qw \\
		& \qw & \ctrl{2} & \qw & \qw \\
		& \qw & \ctrlb{1} & \qw & \qw \\
		& \qw & \gate{X^\bullet} & \qw & \qw
	}
	}
	~~\mp{9}{=}~~
	\scalebox{0.8}{
	\Qcircuit @C=.5em @R=0.2em @!R {
		& \qw & \qw & \qw & \multigate{3}{D} & \ctrl{3} & \qw & \qw
			& \qw  & \qw \\
		& \vdots & & & & & & & \vdots & \\
		& \qw & \qw & \qw & \ghost{D} & \ctrl{2} & \qw
			& \qw & \qw & \qw \\
		& \qw & \qw & \ctrl{2} & \ghost{D} & \qw & \ctrl{2}  & \qw
			& \qw & \qw \\
		& \qw & \gate{H} & \ctrl{1} & \qw & \targ & \ctrl{1} & \gate{H} 
			& \qw & \qw \\
		& \qw & \qw & \targ & \qw & \qw & \targ & \qw & \qw
			& \qw
	}
	}
	~~\mp{9}{=}~~
	\scalebox{0.8}{
	\Qcircuit @C=.5em @R=0.2em @!R {
		& \qw & \qw & \qw & \multigate{3}{D} & \ctrl{3} & \qw  & \qw
			& \ctrl{3} & \multigate{3}{D^\dagger} & \multigate{3}{D} & \ctrl{3} & \qw  & \qw \\
		& \vdots & & & & & & & & & \vdots & \\
		& \qw & \qw & \qw & \ghost{D} & \ctrl{2}  & \qw
			& \qw & \ctrl{2} & \ghost{D^\dagger} & \ghost{D} & \ctrl{2}  & \qw & \qw \\
		& \qw & \qw & \ctrl{2} & \ghost{D} & \qw & \ctrl{2}  & \qw
			& \qw & \ghost{D^\dagger} & \ghost{D} & \qw & \qw & \qw \\
		& \qw & \gate{H} & \ctrl{1} & \qw & \targ & \ctrl{1} & \gate{H}
			& \ctrl{0} & \qw & \qw & \ctrl{0} & \qw & \qw \\
		& \qw & \qw & \targ & \qw & \qw & \targ & \qw  & \qw
			& \qw & \qw & \qw & \qw & \qw
	}
	}
	~~\mp{9}{=}~~
	\scalebox{0.8}{
	\Qcircuit @C=.5em @R=0.2em @!R {
		& \qw & \ctrl{5} & \ctrl{4} & \qw & \qw \\
		& \vdots & & & \vdots & \\
		\push{\rule{0em}{.5em}} & \qw & \ctrl{3} & \ctrl{2} & \qw & \qw \\
		& \qw & \ctrl{2} & \ctrlb{1} & \qw & \qw \\
		& \qw & \qw & \gate{Z^\bullet} & \qw & \qw \\
		& \qw & \targ & \qw & \qw
	}
	}
	\end{align*}
\end{proof}

We now establish the correctness of the circuit constructions in \cref{fig:tla}, and give explicit forms for the relative phases.

\begin{proposition}\label{prop:a2}
\begin{align*}
	\scalebox{0.8}{
		\Qcircuit @C=.5em @R=0em @!R {
		& \qw & \qw & \qw & \qw & \ctrl{4} & \qw & \qw & \qw  \\
		& \vdots & & & & & & \vdots & \\
		& \qw & \qw & \qw & \qw & \qw & \qw & \qw & \qw  \\
		& \qw & \qw & \qw & \qw & \ctrl{2} & \qw & \qw & \qw  \\
		& \qw & \qw & \qw & \qw & \ctrl{1} & \qw & \qw & \qw  \\
		\push{\rule{0em}{1.45em}} & & & &  \lstick{0} & \gate{X^\star} & \gate{S^\dagger} & \qw & \qw  \\
	}
	}
	~~\mp{9}{=}~~
	\scalebox{0.8}{
	\Qcircuit @C=.5em @R=0em @!R {
		& \qw & \ctrl{4} & \ctrl{5} & \qw & \qw  \\
		& \vdots & & & \vdots & \\
		& \qw & \ctrl{2} & \ctrl{3} & \qw & \qw  \\
		& \qw & \ctrlb{1} & \ctrl{2} & \qw & \qw  \\
		& \qw & \gate{Z^\bullet} & \ctrl{1} & \qw & \qw  \\
		\push{\rule{0em}{1.45em}} & & & {} & \qw & \qw  \\
	}
	}
\end{align*}
\end{proposition}
\begin{proof}
\begin{align*}
	\scalebox{0.8}{
		\Qcircuit @C=.5em @R=0em @!R {
		& \qw & \qw & \qw & \qw & \ctrl{4} & \qw & \qw & \qw  \\
		& \vdots & & & & & & \vdots & \\
		& \qw & \qw & \qw & \qw & \ctrl{3} & \qw & \qw & \qw  \\
		& \qw & \qw & \qw & \qw & \ctrl{2} & \qw & \qw & \qw  \\
		& \qw & \qw& \qw & \qw & \ctrl{1} & \qw & \qw & \qw  \\
		\push{\rule{0em}{1.45em}} & & & &  \lstick{0} & \gate{X^\star} & \gate{S^\dagger} & \qw & \qw  \\
	}
	}
	~~&\mp{9}{=}~
	\scalebox{0.8}{
	\Qcircuit @C=.5em @R=0em @!R {
		& \qw & \qw& \qw & \qw & \qw & \qw & \qw & \qw & \ctrl{4} & \qw & \qw 
			& \qw & \qw & \qw & \qw & \qw \\
		& \vdots & & & & & & & & & & & & & & \vdots & \\
		& \qw & \qw & \qw & \qw & \qw & \qw & \qw & \qw & \ctrl{1} & \qw & \qw 
			& \qw & \qw & \qw & \qw & \qw  \\
		& \qw & \qw& \qw & \qw & \qw & \qw & \qw & \qw & \ctrl{2} & \qw & \qw 
			& \qw & \qw & \qw & \qw & \qw \\
		& \qw & \qw& \qw & \qw & \qw & \qw & \ctrl{1} & \qw & \ctrlb{1} & \qw & \ctrl{1} 
			& \qw & \qw & \qw & \qw & \qw \\
		\push{\rule{0em}{1.45em}} & & & &  \lstick{0} & \gate{H} & \gate{T} & \targ & \gate{T^\dagger} 
			& \gate{X^\bullet} & \gate{T} & \targ & \gate{T^\dagger} 
			& \gate{H} & \gate{S^\dagger} & \qw & \qw
	}
	} \\ 
	~~&\mp{9}{=}~
	\scalebox{0.8}{
	\Qcircuit @C=.5em @R=0em @!R {
		& \qw & \qw & \qw & \qw & \ctrl{3} & \qw & \qw & \qw & \ctrl{4} & \qw & \qw 
			& \qw & \qw & \qw & \qw & \qw \\
		& \vdots & & & & & & & & & & & & & & \vdots & \\
		& \qw & \qw & \qw & \qw & \ctrl{2} & \qw  & \qw & \qw & \ctrl{1} & \qw & \qw 
			& \qw & \qw & \qw & \qw & \qw  \\
		& \qw & \qw & \qw & \qw & \ctrlb{1} & \qw & \qw & \qw & \ctrl{2} & \qw & \qw 
			& \qw & \qw & \qw & \qw & \qw \\
		& \qw & \qw & \qw & \qw & \gate{Z^\bullet}  & \qw & \ctrl{1} & \qw & \qw & \qw & \ctrl{1} 
			& \qw & \qw & \qw & \qw & \qw \\
		\push{\rule{0em}{1.45em}} & & & &  \lstick{0} & \gate{H} & \gate{T} & \targ & \gate{T^\dagger} 
			& \targ & \gate{T} & \targ & \gate{T^\dagger} 
			& \gate{H} & \gate{S^\dagger} & \qw & \qw
	}
	} \\ 
	~~&\mp{9}{=}~
	\scalebox{0.8}{
		\Qcircuit @C=.5em @R=0em @!R {
		& \qw & \qw & \qw & \qw & \ctrl{3} & \ctrl{3} & \ctrl{4} & \qw & \qw & \qw  \\
		& \vdots & & & & & & & & \vdots & \\
		& \qw & \qw & \qw & \qw & \ctrl{2} & \ctrl{2} & \ctrl{1} & \qw & \qw & \qw  \\
		& \qw & \qw & \qw & \qw & \ctrlb{1} & \ctrl{1} & \ctrl{2} & \qw & \qw & \qw  \\
		& \qw & \qw & \qw & \qw & \gate{Z^\bullet} & \gate{S} & \ctrl{1} & \qw & \qw & \qw  \\
		\push{\rule{0em}{1.45em}} & & & & \lstick{0} & \qw & \qw & \targ & \gate{S^\dagger} & \qw & \qw  \\
	}
	} \\ 
	~~&\mp{9}{=}~
	\scalebox{0.8}{
		\Qcircuit @C=.5em @R=0em @!R{
		& \qw & \qw & \ctrl{4} & \ctrl{5} & \qw & \qw  \\
		& \vdots & & & & \vdots & \\
		& \qw & \qw & \ctrl{2} & \ctrl{3} & \qw & \qw  \\
		& \qw & \qw & \ctrlb{1} & \ctrl{2} & \qw & \qw  \\
		& \qw & \qw & \gate{Z^\bullet} & \ctrl{1} & \qw & \qw  \\
		\push{\rule{0em}{1.45em}} & & &  \lstick{0} & \targ & \qw & \qw  \\
	}
	}
	~~\mp{9}{=}~~
	\scalebox{0.8}{
	\Qcircuit @C=.5em @R=0em @!R {
		& \qw & \ctrl{4} & \ctrl{5} & \qw & \qw  \\
		& \vdots & & & \vdots & \\
		& \qw & \ctrl{2} & \ctrl{3} & \qw & \qw  \\
		& \qw & \ctrlb{1} & \ctrl{2} & \qw & \qw  \\
		& \qw & \gate{Z^\bullet} & \ctrl{1} & \qw & \qw  \\
		\push{\rule{0em}{1.45em}} & & & {} & \qw & \qw  \\
	}
	}
\end{align*}
\end{proof}

\begin{proposition}\label{prop:a3}
\vspace*{-1em}
\begin{align*}
	\scalebox{0.8}{
		\Qcircuit @C=.5em @R=0em @!R {
		& \qw & \qw & \qw & \qw & \qw & \qw & \qw & \qw & \qw & \ctrl{5} 
			& \qw & \qw & \qw & \qw & \qw & \qw  \\
		& \vdots & & & & & & & & & & & & & & \vdots & \\
		& \qw & \qw & \qw & \qw & \qw & \qw & \qw & \qw & \qw & \ctrl{3} 
			& \qw & \qw & \qw & \qw & \qw & \qw  \\
		& \qw & \qw & \qw & \qw & \qw & \qw & \qw & \ctrl{2} & \qw & \ctrlb{2} 
			& \qw & \ctrl{2} & \qw & \qw & \qw & \qw  \\
		& \qw & \qw & \qw & \qw & \qw & \qw & \qw & \ctrl{1} & \ctrl{1} & \qw 
			& \ctrl{1} & \ctrl{1} & \qw & \qw & \qw & \qw  \\
		& \qw & \gate{H} & \meter\cwx[1]  & & & {|} & \qw & \gate{iX} 
			& \targ & \gate{{}^\bullet Z} & \targ & \gate{-iX} & \qw & {|}\qw & \\
		\push{\rule{0em}{1.45em}} & \cw & \cw & \cctrl{0} & \cw & \cw & \cw & \cw & \cctrl{-1} & \cw 
			& \cw & \cw & \cctrl{-1} & \cw & \cw & \cw & \cw
	}
	}
	~\mp{9}{=}~
	\scalebox{0.8}{
		\Qcircuit @C=.5em @R=0em @!R {
		& \qw & \ctrl{5} & \ctrl{3} & \qw & \qw  \\
		& \vdots & & & \vdots & \\
		& \qw & \ctrl{3} & \ctrl{2} & \qw & \qw  \\
		& \qw & \ctrl{2} & \ctrlb{1} & \qw & \qw  \\
		& \qw & \ctrl{1} &  \gate{{}^\bullet Z} & \qw & \qw  \\
		\push{\rule{0em}{1.45em}} & \qw & {}\qw  & & &  \\
		& & & &
	}
	}
\end{align*}
\vspace*{-1em}
\end{proposition}
\begin{proof}
\vspace*{-1em}
\begin{align*}
	\scalebox{0.8}{
		\Qcircuit @C=.5em @R=0em @!R {
		& \qw & \qw & \qw & \qw & \qw & \qw & \qw & \qw & \qw & \ctrl{5} 
			& \qw & \qw & \qw & \qw & \qw & \qw  \\
		& \vdots & & & & & & & & & & & & & & \vdots & \\
		& \qw & \qw & \qw & \qw & \qw & \qw & \qw & \qw & \qw & \ctrl{3} 
			& \qw & \qw & \qw & \qw & \qw & \qw  \\
		& \qw & \qw & \qw & \qw & \qw & \qw & \qw & \ctrl{2} & \qw & \ctrlb{2} 
			& \qw & \ctrl{2} & \qw & \qw & \qw & \qw  \\
		& \qw & \qw & \qw & \qw & \qw & \qw & \qw & \ctrl{1} & \ctrl{1} & \qw 
			& \ctrl{1} & \ctrl{1} & \qw & \qw & \qw & \qw  \\
		& \qw & \gate{H} & \meter\cwx[1]  & & & {|} & \qw & \gate{iX} 
			& \targ & \gate{{}^\bullet Z} & \targ & \gate{-iX} & \qw & {|}\qw & \\
		\push{\rule{0em}{1.45em}} & \cw & \cw & \cctrl{0} & \cw & \cw & \cw & \cw & \cctrl{-1} & \cw 
			& \cw & \cw & \cctrl{-1} & \cw & \cw & \cw & \cw
	}
	}
	~&\mp{9}{=}~
	\scalebox{0.8}{
		\Qcircuit @C=.5em @R=0em @!R {
		& \qw & \qw & \qw & \qw & \qw & \qw & \qw & \qw & \qw & \ctrl{5} 
			& \qw & \qw & \qw & \qw & \qw & \qw  \\
		& \vdots & & & & & & & & & & & & & & \vdots & \\
		& \qw & \qw & \qw & \qw & \qw & \qw & \qw & \qw & \qw & \ctrl{3} 
			& \qw & \qw & \qw & \qw & \qw & \qw  \\
		& \qw & \qw & \qw & \qw & \qw & \qw & \qw & \ctrl{2} & \qw & \ctrlb{2} 
			& \qw & \ctrl{2} & \qw & \qw & \qw & \qw  \\
		& \qw & \qw & \qw & \qw & \qw & \qw & \qw & \ctrl{1} & \ctrl{1} & \qw 
			& \ctrl{1} & \ctrl{1} & \qw & \qw & \qw & \qw  \\
		& & & & & & {|} & \qw & \gate{iX} 
			& \targ & \gate{{}^\bullet Z} & \targ & \gate{-iX} & \qw & {|}\qw & \\
		\push{\rule{0em}{1.45em}} & \qw & \gate{H} & \meter  & \cw & \cw & \cw & \cw & \cctrl{-1} & \cw 
			& \cw & \cw & \cctrl{-1} & \cw & \cw & \cw & \cw
	}
	} \\
	~&\mp{9}{=}~
	\scalebox{0.8}{
		\Qcircuit @C=.5em @R=0em @!R {
		& \qw & \qw & \qw & \qw & \qw & \qw & \qw & \qw & \qw & \ctrl{5} 
			& \qw & \qw & \qw & \qw & \qw & \qw  \\
		& \vdots & & & & & & & & & & & & & & \vdots & \\
		& \qw & \qw & \qw & \qw & \qw & \qw & \qw & \qw & \qw & \ctrl{3} 
			& \qw & \qw & \qw & \qw & \qw & \qw  \\
		& \qw & \qw & \qw & \qw & \qw & \qw & \qw & \ctrl{2} & \qw & \ctrlb{2} 
			& \qw & \ctrl{2} & \qw & \qw & \qw & \qw  \\
		& \qw & \qw & \qw & \qw & \qw & \qw & \qw & \ctrl{1} & \ctrl{1} & \qw 
			& \ctrl{1} & \ctrl{1} & \qw & \qw & \qw & \qw  \\
		& & & & & & {|} & \qw & \gate{iX} 
			& \targ & \gate{{}^\bullet Z} & \targ & \gate{-iX} & \qw & {|}\qw & \\
		\push{\rule{0em}{1.45em}} & \qw & \gate{H} & \qw & \qw & \qw & \qw & \qw & \ctrl{-1} & \qw 
			& \qw & \qw & \ctrl{-1} & \qw & \qw & \meter  
	}
	} \\
	~&\mp{9}{=}~
	\scalebox{0.8}{
		\Qcircuit @C=.5em @R=0em @!R {
		& \qw & \qw & \qw & \qw & \qw & \qw & \qw & \qw & \qw & \ctrl{5} 
			& \ctrl{2} & \qw & \qw & \qw & \qw & \qw  \\
		& \vdots & & & & & & & & & & & & & & \vdots & \\
		& \qw & \qw & \qw & \qw & \qw & \qw & \qw & \qw & \qw & \ctrl{3} 
			& \ctrlb{1} & \qw & \qw & \qw & \qw & \qw  \\
		& \qw & \qw & \qw & \qw & \qw & \qw & \qw & \ctrl{2} & \qw & \qw 
			& \gate{{}^\bullet Z} & \ctrl{2} & \qw & \qw & \qw & \qw  \\
		& \qw & \qw & \qw & \qw & \qw & \qw & \qw & \ctrl{1} & \ctrl{1} & \qw 
			& \ctrl{1} & \ctrl{1} & \qw & \qw & \qw & \qw  \\
		& & & & & & {|} & \qw & \gate{iX}
			& \targ & \ctrl{0} & \targ & \gate{-iX} & \qw & {|}\qw & \\
		\push{\rule{0em}{1.45em}} & \qw & \gate{H} & \qw & \qw & \qw & \qw & \qw & \ctrl{-1} & \qw 
			& \qw & \qw & \ctrl{-1} & \qw & \qw & \meter  
	}
	} \\
	~&\mp{9}{=}~
	\scalebox{0.8}{
		\Qcircuit @C=.5em @R=0em @!R {
		& \qw & \qw & \qw & \qw & \qw & \qw & \qw & \qw & \ctrl{4} & \ctrl{2} & \ctrl{5} 
			& \qw & \qw & \qw & \qw & \qw  \\
		& \vdots & & & & & & & & & & & & & & \vdots & \\
		& \qw & \qw & \qw & \qw & \qw & \qw & \qw & \qw & \ctrl{2} & \ctrlb{1} & \ctrl{3} 
			& \qw & \qw & \qw & \qw & \qw  \\
		& \qw & \qw & \qw & \qw & \qw & \qw & \qw & \ctrl{2} & \qw & \gate{{}^\bullet Z} & \qw
			& \ctrl{2} & \qw & \qw & \qw & \qw  \\
		& \qw & \qw & \qw & \qw & \qw & \qw & \qw & \ctrl{1} & \ctrl{0} & \qw & \qw
			& \ctrl{1} & \qw & \qw & \qw & \qw  \\
		& & & & & & {|} & \qw & \gate{iX} 
			& \qw & \qw & \ctrl{0} & \gate{-iX} & \qw & {|}\qw & \\
		\push{\rule{0em}{1.45em}} & \qw & \gate{H} & \qw & \qw & \qw & \qw & \qw & \ctrl{-1} & \qw 
			& \qw & \qw & \ctrl{-1} & \qw & \meter  
	}
	} \\
	~&\mp{9}{=}~
	\scalebox{0.8}{
		\Qcircuit @C=.5em @R=0em @!R {
		& \qw & \qw & \qw & \qw & \qw & \qw & \qw & \ctrl{4} & \ctrl{5} 
			& \qw & \qw & \qw & \qw  \\
		& \vdots & & & & & & & & & & & & \vdots & \\
		& \qw & \qw & \qw & \qw & \qw & \qw & \qw & \ctrl{2} & \ctrl{3} 
			& \qw & \qw & \qw & \qw  \\
		& \qw & \qw & \qw & \qw & \qw & \qw & \qw & \ctrlb{1} & \ctrl{2}
			& \qw & \qw & \qw & \qw  \\
		& \qw & \qw & \qw & \qw & \qw & \qw & \qw & \gate{{}^\bullet Z} & \ctrl{1} 
			& \qw & \qw & \qw & \qw  \\
		\push{\rule{0em}{1.45em}} & \qw & \gate{H} & \qw & \qw & \qw & \qw & \qw & \qw 
			& \ctrl{0} & \qw & \qw & \meter  
	}
	}\\
	~&\mp{9}{=}~
	\scalebox{0.8}{
		\Qcircuit @C=.5em @R=0em @!R {
		& \qw & \ctrl{5} & \ctrl{4} & \qw & \qw  \\
		& \vdots & & & & \vdots & \\
		& \qw & \ctrl{3} & \ctrl{2} & \qw & \qw  \\
		& \qw & \ctrl{2} & \ctrlb{1} & \qw & \qw  \\
		& \qw & \ctrl{1} & \gate{{}^\bullet Z} & \qw & \qw  \\
		\push{\rule{0em}{1.45em}} & \qw & \targ & \qw & \lstick{0}  \\
	}
	}
	~\mp{9}{=}~
	\scalebox{0.8}{
		\Qcircuit @C=.5em @R=0em @!R {
		& \qw & \ctrl{5} & \ctrl{3} & \qw & \qw  \\
		& \vdots & & & \vdots & \\
		& \qw & \ctrl{3} & \ctrl{2} & \qw & \qw  \\
		& \qw & \ctrl{2} & \ctrlb{1} & \qw & \qw  \\
		& \qw & \ctrl{1} &  \gate{{}^\bullet Z} & \qw & \qw  \\
		\push{\rule{0em}{1.45em}} & \qw & {}\qw  & & &  \\
		& & & &
	}
	}
\end{align*}
\vspace*{-1em}
\end{proof}

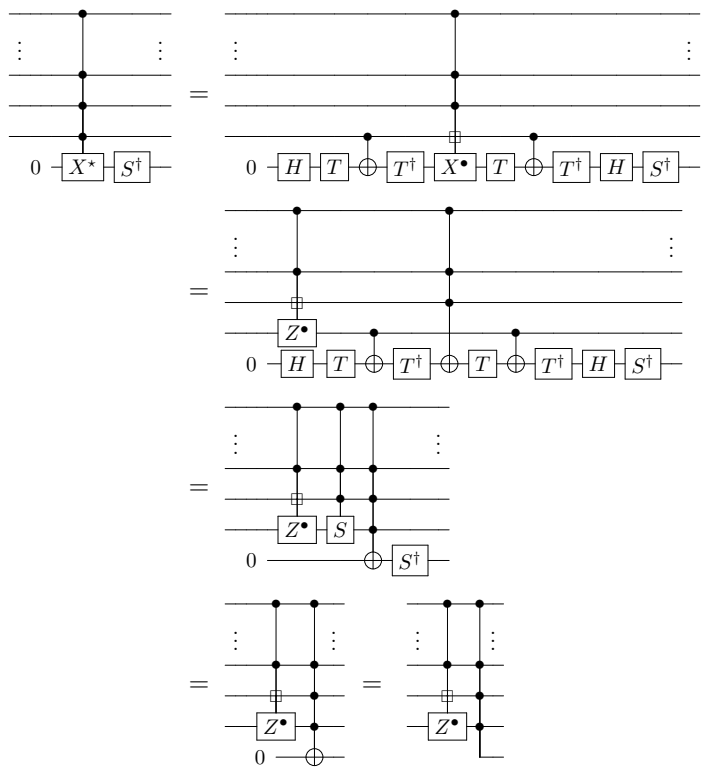
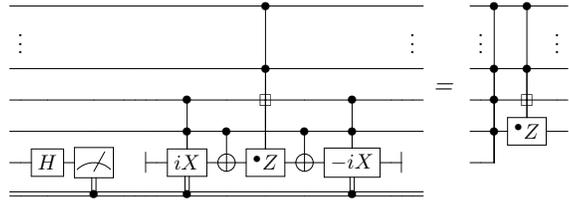
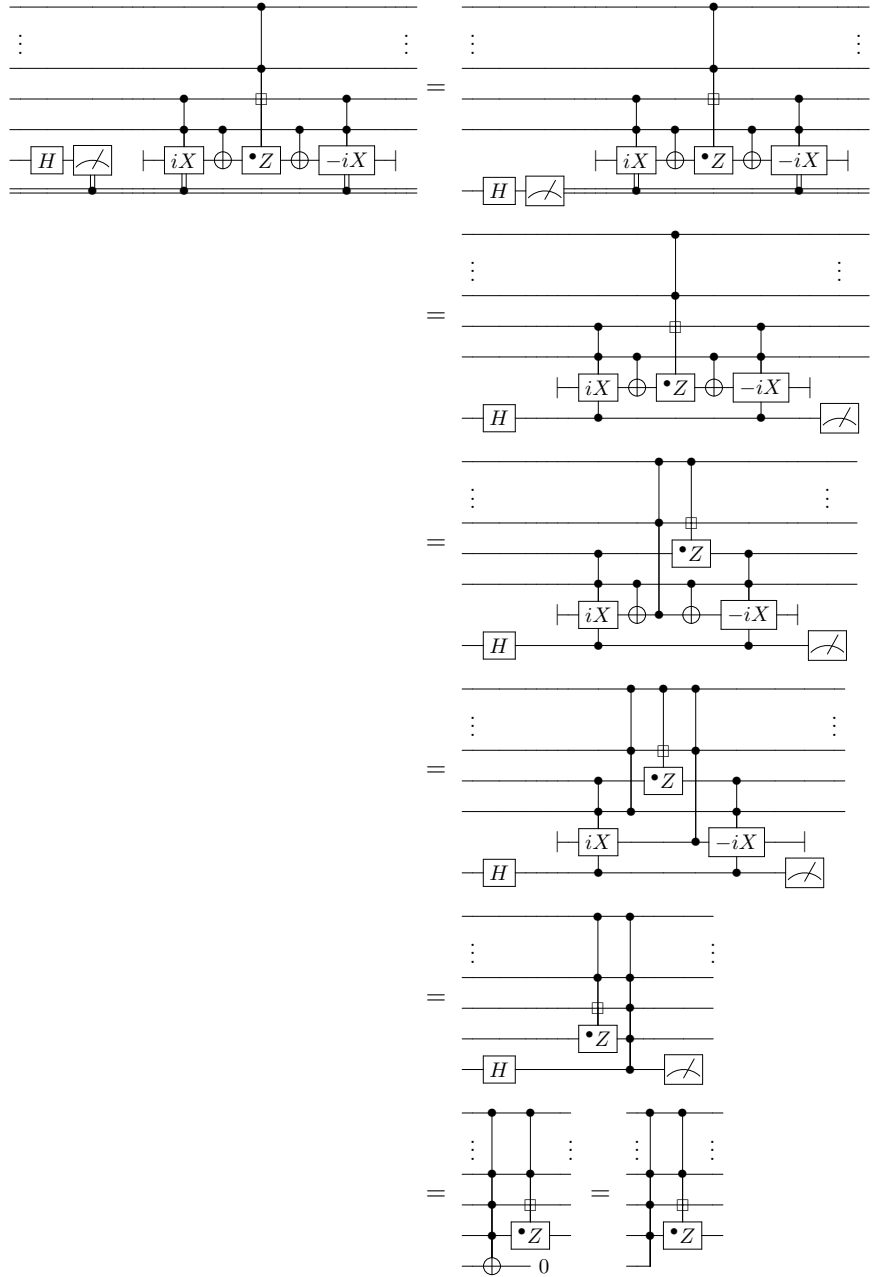
\end{document}